%% file: poolSepa.tex
\documentclass[10pt,conference,letterpaper]{IEEEtran}

\usepackage{amssymb,url,amsmath, relsize, accents}
\usepackage{mathrsfs, enumerate}
\usepackage{subfigure}
\usepackage{graphicx}
\usepackage{multicol}
\usepackage{galois}
\usepackage{amsthm}
\usepackage[T1]{fontenc}
\usepackage[utf8]{inputenc}
\usepackage{authblk}

\def\expect{{\mathbb  E}}
\def\Pr{{\mathbb P}}
\def\RR{{\mathbb R}}

\def\eqdef{\triangleq}

\newtheorem{theorem}{Theorem}

\newtheorem{lemma}{Lemma}
\newtheorem{corollary}{Corollary}

\newtheorem{definition}{Definition}

\newcounter{example}[section]

  \newcounter{remark}[section]

 \newcounter{experiment}[section]
\newenvironment{experiment}[1][]{\refstepcounter{experiment}\par\medskip
   \noindent \textbf{Experiment~\theexperiment. #1} \rmfamily}{\medskip}

\begin{document}

\date{}

\title{On Resource Pooling and Separation for LRU Caching}

\author{Jian Tan,\;\;Guocong Quan,\;\;Kaiyi Ji,\;\;Ness Shroff\\
            The Ohio State University\\
            Columbus, OH 43210\\
            \{tan.252, quan.72, ji.367, shroff.11\}@osu.edu}

\maketitle

\begin{abstract}
Caching systems using the Least Recently Used (LRU) principle have now become ubiquitous. A fundamental question for these systems is whether the cache space should be pooled together or divided
to serve multiple flows of data item requests
in order to minimize the miss probabilities. In this paper, we show that there is no 
straight yes or no answer to this question, depending on complex combinations 
of critical factors, 
including, e.g.,  request rates, overlapped data items
across different request flows, data item popularities and their sizes. 
Specifically,  we characterize the asymptotic miss probabilities for multiple competing request flows under
resource pooling and separation for LRU caching when the cache size is large.  

Analytically, we show that it is asymptotically optimal to jointly serve multiple flows
 if their data item sizes and popularity distributions are similar and their arrival rates 
do not differ significantly; the self-organizing property of LRU caching automatically optimizes the resource allocation among them asymptotically. 
Otherwise, separating these flows could be better, e.g., when data sizes vary significantly. 
We also quantify critical points beyond which resource pooling is better than separation for each of the flows when the overlapped data items exceed certain levels.
Technically, we generalize existing results on the asymptotic miss probability of LRU caching for a 
broad class of heavy-tailed distributions and extend them to 
multiple competing flows with varying data item sizes, which also validates the Che approximation under certain conditions.  
These results provide new insights on improving the performance of caching systems. 
\end{abstract}

\section{Introduction}\label{s:intro}
Caching systems using the Least Recently Used (LRU) principle are already widely deployed 
but need to efficiently scale to support emerging data applications.
They have very different stochastic dynamics~\cite{che2002,Rosensweig:2010,Fricker:2012,Choungmo2014,BergerGSC14,jung2003modeling,Gast:2015,leaseMiss,Berger:2015:mama,Ferragut:2016} than well-studied queueing systems.  
One cannot apply the typical intuition of resource pooling for queueing, e.g.,~\cite{kella89,Harrison1999,bell2001,stolyar2004, borgs2013},  to caching.   
To serve multiple flows of data item requests,
a fundamental question is 
whether the cache space should be pooled together or divided (see Fig.~\ref{fig:singleServer})
in order to minimize the miss probabilities.
\begin{figure}[h]\vspace{-0.3cm}
\centering
\includegraphics[width=4.5cm]{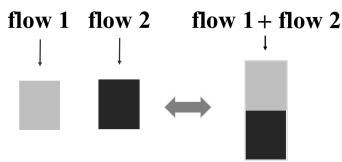}\vspace{-0.27cm}
\caption{Flows served separately and jointly}\label{fig:singleServer}
\vspace{-0.2cm}
\end{figure}

A request is said to ``miss'' if the corresponding data item is not found in the cache; otherwise a ``hit'' occurs.  
For a web service each miss often incurs subsequent work at a backend database,  resulting in overhead as high as a few milliseconds or even seconds~\cite{zExpander}.   
A study on Facebook's memcached workloads shows that
a small percentage of miss ratio on one server can trigger 
millions of requests to
the database per day~\cite{JiangSong,cacheDelay}. 
Thus, even a minor increase in the hit ratio can significantly improve system performance. 
To further motivate the problem, we examine the cache space
 allocation for in-memory key-value storage systems. 

\subsection{Background and current practice}

In-memory cache processing can
greatly expedite data retrieval, since data are kept in Random Access Memory (RAM).  
In a typical key-value cache system, e.g., Memcached~\cite{memcached,memc3:nsdi13,memcachNSDI}, a data item is added 
to the cache after
a client has requested it and failed.  When the cache is full,  an old data item needs to be evicted to 
make room for the new one. This selection is determined by the caching algorithm. 
Different caching algorithms have been proposed~\cite{LIRS,ARC}.
However, due to the cost of tracking access history,  
often only LRU or its approximations~\cite{AndrewOS},  are adopted~\cite{JiangSong}. 
The LRU algorithm replaces the data item that has not been used for the longest period of time. 

The current engineering practice is to organize servers into pools
based on applications and data domains~\cite{JiangSong,memcachNSDI,TAO}. On a server, the cache space
is divided into isolated slabs according to data item sizes~\cite{memcached, cacheDelay}.  
Note that different servers and slabs have separate LRU lists.  
These solutions have 
yielded good performance~\cite{memcached, cidon2016cliffhanger, cacheDelay}, 
through coarse level control on resource pooling and separation. 
However, it is not clear whether these rules of thumb are optimal allocations, or whether 
 one can develop simple solutions to further improve the performance.

\subsection{The optimal strategy puzzle}
These facts present a dilemma.
On the one hand,  multiple request flows benefit from resource pooling. For example,  a shared cache space
that provides sufficiently high hit ratios for two flows 
can improve the utilization of the limited RAM space, especially when 
the two flows contain overlapped common data items so that a data item 
brought into cache by one flow can be directly used by others. 
On the other hand, resource separation facilitates capacity planning 
for different flows and ensures adequate quality of service
for each.  For example, a dedicated cache space can prevent one flow with a high
request rate from evicting too many data items of another competing flow on the same cache~\cite{JiangSong}.   

This dilemma only scratches the surface of whether resource pooling or separation is better for caching. Four critical factors complicate the problem and jointly impact the cache miss probabilities (a.k.a. miss ratios), including request rates, overlapped data items
across different request flows, data item popularities and their sizes.  Depending on the setting, they may lead to different conclusions.  Below we demonstrate the 
complexity of the optimal strategy using three examples, showing that
resource pooling can be asymptotically equal to, better or worse than separation, respectively.  
Consider two independent flows ($1$ and~$2$) of requests with Poisson arrivals of rates $\nu_1$ and $\nu_2$, respectively. 
The data items of the two flows do not overlap and have unit sizes unless explicitly specified. Their popularities follow truncated Zipf's distributions,  $p^{(1)}_i=c_1/i^{\alpha_1}$ and $p^{(2)}_j=c_2/j^{\alpha_2}, 1\leq i,j \leq N$, where $i,j$ are the indeces of the data items
of flow~$1$ and $2$, respectively.   
For pooling, two flows share the whole cache.
For separation, the cache is partitioned into two parts using fractions $u_1$ and $u_2$, to serve flow~$1$ and~$2$ separately, $u_1 + u_2=1, u_1, u_2 \geq 0$.

\vspace{0.15cm}
\noindent {\bf Case 1: Asymptotic equivalence}

\noindent
The optimal resource separation scheme
has recently been shown to be better than pooling~\cite{chu2016allocating} under certain assumptions based on the Che approximation~\cite{che2002}. 
However, it is not clear whether the difference is significant or not, especially when the cache size is large (a typical scenario). 
The first example shows that they can be quite close. Notably resource pooling is adaptive and need not optimize separation fractions $u_1$.  
For $\alpha_1=1.5, \alpha_2=4.0, \nu_1=0.1, \nu_2=0.9, N=10^6$, we plot the overall miss probabilities under resource pooling and separation
in Fig.~\ref{fig:distribution1}, respectively. The optimal ratio $u_1$ for separation is obtained numerically by an exhaustive search.  
\begin{figure}[h] 
\vspace{-0.3cm}
\centering
\includegraphics[width=6.8cm]{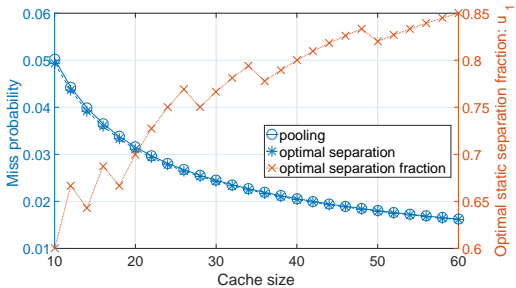}\vspace{-0.3cm}
\caption{Asymptotically equal miss ratios} \label{fig:distribution1}
\vspace{-0.3cm}
\end{figure}
When the cache size is small, the optimal separation strategy achieves a better miss probability than resource pooling.
However, for large cache sizes, the miss probabilities are indistinguishable. This is not an coincidence, as shown by Theorem~\ref{lemma:samealpha}.  
Note that the cache sizes take integer values, thus $u_1$ varying up and down. 

\vspace{0.15cm}
\noindent {\bf  Case 2: Pooling is better} 

\noindent
The previous example shows that resource pooling can adaptively achieve the best separation fraction when the cache space is large. Consider two flows with $\alpha_1=\alpha_2=2, N=10^6$ and time-varying Poisson request rates. 
For $T=10^6$, let $\nu_1=0.1, \nu_2=0.9$ in the time interval $[(2k-1)T+1,2kT]$ and $\nu_1=0.9,\nu_2=0.1$ in $[2kT+1,(2k+1)T]$, $k=1,2,\cdots$.
\begin{figure}[h] 
\vspace{-0.3cm}
\centering
\includegraphics[width=6.8cm]{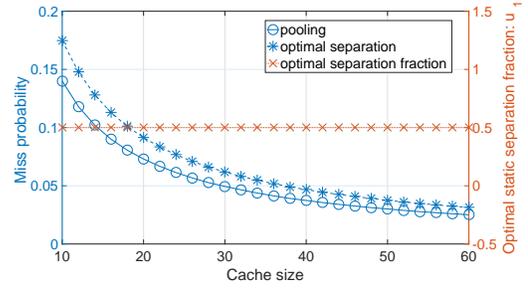}\vspace{-0.3cm}
\caption{Benefits of pooling due to self-organization} \label{fig:requestRate}
\vspace{-0.3cm}
\end{figure} 
The simulation results in Fig~\ref{fig:requestRate} show that resource pooling achieves a smaller miss probability, primarily attributing to self-organization. The optimal static separation ratio in this case is $u_1=0.5$ due to symmetry.

\noindent {\bf Case 3: Separation is better}

\noindent
Assume that the data items from flow~$1$ and flow~$2$ have different sizes $1$ and $4$, respectively, with $N=10^6, \alpha_1=\alpha_2=2, \nu_1=\nu_2=0.5$.
The simulation results in Fig.~\ref{fig:diffSize} show that the optimal separation yields a better performance due to 
varying data item sizes, which is supported by Theorem~\ref{lemma:samealpha}.
This may explain why in practice it is beneficial to separate cache space according to
applications, e.g.,  text and image objects, which could have significantly different item sizes~\cite{memcachNSDI, xu2014characterizing}.
What if the data item sizes are equal? Fig.~\ref{fig:distribution1} is an example that separation is better when 
\begin{figure}[h] 
\vspace{-0.3cm}
\centering
\includegraphics[width=6.8cm]{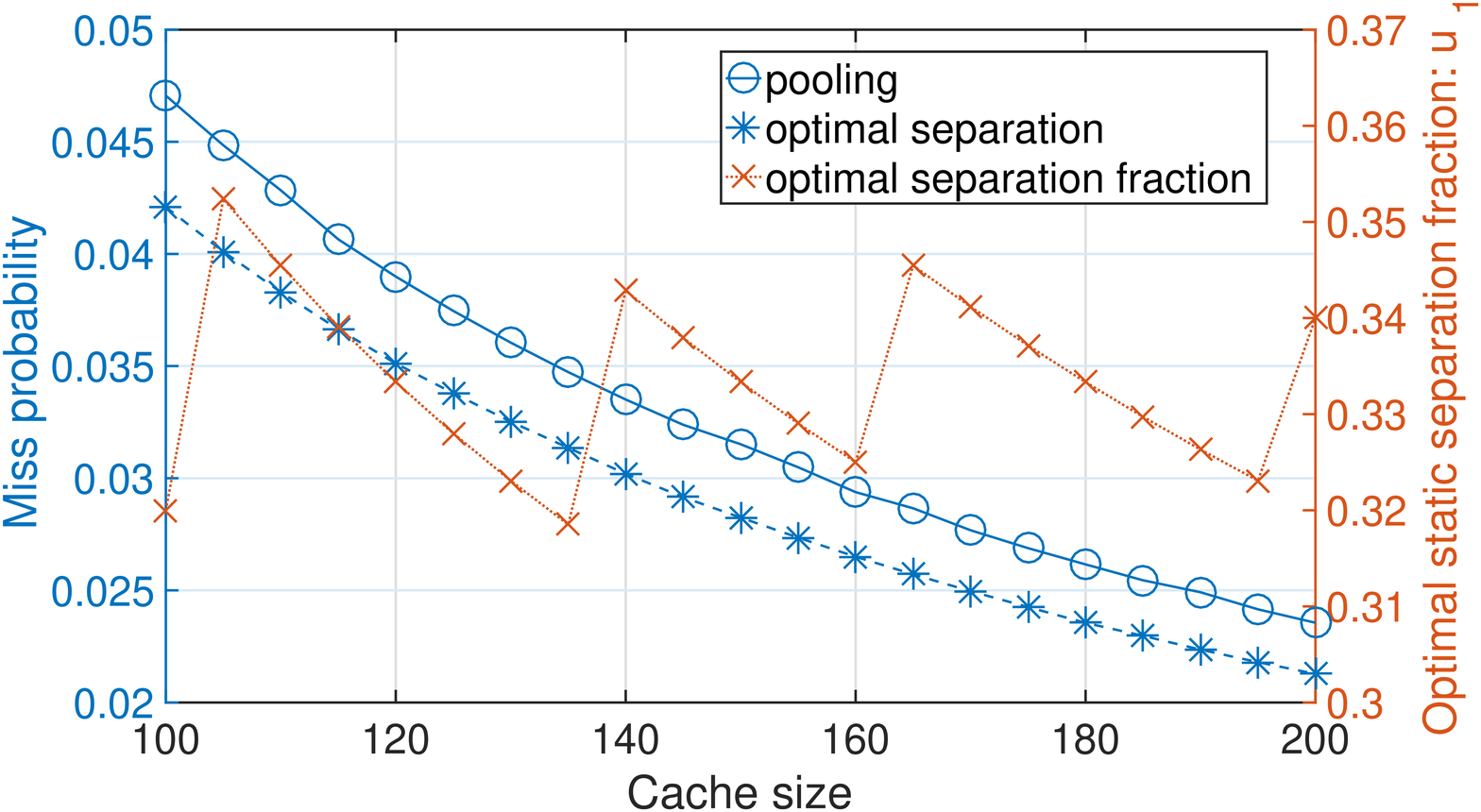}\vspace{-0.3cm}
\caption{Benefits of separation due to isolation} \label{fig:diffSize}
\vspace{-0.3cm}
\end{figure}
the cache space is small even with equal data 
item sizes. However, a small cache may not be typical for caching systems.  
These examples motivate us to systematically study the miss probabilities for competing flows with different rates,  distributions, and partially overlapped data items of varying sizes. Our analytical results can be used to explain the puzzling performance differences demonstrated in the previous three examples.

\subsection{Summary of contributions}

\noindent (1) An analytical framework under the independent reference model (IRM)~\cite{Coffman} is proposed to address four critical factors for LRU caching: request rates, distributions, data item sizes and the overlapped data items across different flows. 
We generalize the existing results~\cite{Jelenkovic99asymptoticapproximation,Jelenkovic:2004} on the asymptotic miss probability of LRU caching from Zipf's law to a broad class of heavy-tailed distributions, including, e.g.,  regularly varying and heavy-tailed Weibull distributions.  More importantly, our results can characterize miss probabilities of 
multiple competing flows with varying data item sizes when they share a common large cache space.  These asymptotic results validate the Che approximation~\cite{che2002} under certain conditions. 

\noindent (2) Based on the miss probabilities for both the aggregated and the individual flows, we 
provide guidance on whether multiple competing flows should be served together or not.  
First, we show that when the flows have similar distributions and equal data item sizes, 
the self-organizing property of LRU can adaptively search for the optimal resource allocation for shared flows. 
As a result, the overall miss probability of the aggregated flows is asymptotically equal to the miss probability using 
the optimal static separation scheme. 
In addition, 
if the request rates of these flows are close,  the miss probabilities of individual flows when served jointly differ only by a small constant factor 
compared to the case when they are served separately.  
Otherwise, either some of the request flows will be severely penalized or the total miss ratio will become worse.  In that case, it is better to separately serve them.
Second, we consider multiple flows with overlapped data.  When the overlapped data items exceed a certain level,  there exists a region such that every flow can get a better hit ratio. However, 
if not in this region, e.g., when the arrival rates are very different, some flows will be negatively impacted by other competing flows. Based on the analysis, we discuss engineering implications. 

\noindent (3) Extensive simulations are conducted to verify the theoretical results. We design a number of simulations, with different purposes and emphases, 
and show an accurate match with our theoretical results.

\subsection{Related work}

LRU caching is a self-organizing list~\cite{allenM78,albersW98,bitner,blumCK,EstivillW,gonnetMS,rivest,hesterH,arCT} that has been extensively studied. There are two basic approaches to conduct the analysis: combinatorial and probabilistic. The first approach focuses on the classic amortized~\cite{bentleyM,borodinE,RZ,sleatorT85,sleatorTpage} and competitive analysis~\cite{Knuth:1998,Bentley:1985,Chung:1985,aroraKK,Borodin:1991,Irani:1992}. The second approach includes average case analysis~\cite{albersM98,motwani94,AlonS} and stochastic 
analysis~\cite{nelson1977single,fill:1996, FILL1996185,Flajolet:1992,Dobrow1995,barrera2004}.   
When cache sizes are small, the miss probabilities can be explicitly computed~\cite{aven1976some, aven1987stochastic,babaoglu1983two,gelenbe1973unified}.
For large cache sizes, a number of works (e.g.,~\cite{ berger2014exact, fofack2012analysis,garetto2016unified,jung2003modeling,roberts2013exploring})
rely on the Che approximation~\cite{che2002}, which has been extended to cache networks~\cite{laoutaris2006lcd, fofack2012analysis, rosensweig2010approximate, gallo2014performance, garetto2016unified,berger2014exact}.
For fluid limits as scaling factors go to infinity (large cache sizes), mean field 
approximations of the miss probabilities have been developed~\cite{hirade2010analysis, tsukada2012fluid,gast2015transient}.
For emerging data processing systems, e.g., Memcached~\cite{memcached}, since the cache sizes are usually large and the miss probabilities are controlled to be small,  it is natural to conduct the asymptotic analysis of the miss probabilities~\cite{Jelenkovic99asymptoticapproximation,Jelenkovic:2004}. 
Although the miss ratios are small, they still significantly impact the caching system performance.  
Nevertheless, most existing 
works do not address multiple competing request flows on a shared cache space, which can impact each other through complicated ways. 

Workload measurements for caching systems~\cite{arlitt,manley97,Cunha:1995,arlitt,lee99,bradley99,Arlitt:1999:WCW,JiangSong} 
are the basis for theoretical modeling and system optimization.  
Empirical trace studies show that many characteristics of Web caches can be modeled using power-law distributions~\cite{arlitt,Yang:2016}, including, e.g.,  the overall
data item popularity rank, the document sizes, the
distribution of user requests for documents~\cite{lee99,manley97,barford99,Arlitt:1999:WCW}, and the write traffic~\cite{Yang:2016}. Similar phenomena have also been found for large-scale key-value stores~\cite{JiangSong}. 
These facts motivate us to exploit the heavy-tailed workload characteristics.   

Web and network caching is closely related to this study with a large 
body of dedicated works; see the surveys~\cite{Wang:1999:SWC,Podlipnig:2003:SWC} and the references therein.    
Recently a utility optimization approach~\cite{chu2016allocating,dehghan2017sharing} based on the Che approximation~\cite{che2002,BergerGSC14} 
has been used to study
cache sharing and partitioning.
It has concluded that under certain settings the optimal resource separation 
is better than pooling. 
However, it is not clear whether the difference is significant or not, especially when the cache size is large for a typical scenario. We show that 
a simple LRU pooling is asymptotically equivalent to the optimal separation scheme for certain settings, which is significant since the former
is adaptive and does not require any configuration or tuning optimization.
We focus on the asymptotic miss probabilities for multiple competing flows directly, as the miss ratio is one of the most important metrics 
for caching systems with large cache sizes in practice.

\section{Model and intuitive results}\label{s:model}

Consider $M$ flows of i.i.d. random data item requests that are mutually independent. 
Assume that the arrivals of flow~$k$ follow a Poisson process with rate $\lambda_k>0, 1\leq k \leq M$.   The arrivals of the mixed $M$ request flows occur at time points  $\{\tau_n, -\infty<n<+\infty \}$.  
Let $I_n$ be the index of the flow for the request at $\tau_n$. The event $\{I_n=k\}$ represents that the request at $\tau_n$ originates from flow~$k$.  Due to the Poisson assumption, we have $\Pr[I_n=k] = \lambda_k/\left(\sum_{i}\lambda_i\right)$.

To model the typical scenario that the number of distinct data items far exceeds the cache capacity,  we assume that each flow can access an infinite number of data items. 
Formally, flow~$k$ accesses the set of data items $d^{(k)}_i, i=1,2,\cdots, \infty, 1\leq k \leq M$,  from which only  
a finite number can be stored in cache due to the limited capacity. Let $s^{(k)}_i$ denote the size of data item $d^{(k)}_i$.   Note that it is possible, and even common in practice, to observe $d^{(k)}_i \equiv d^{(g)}_j$ for flows $k$ and $g$, where ``$\equiv$'' means that
the two involved data items are the same. 
Therefore, this model describes the situation when data items can overlap between different flows.  
 \begin{figure}[h!]\vspace{-0.4cm}
\centering
\includegraphics[width=8.1cm]{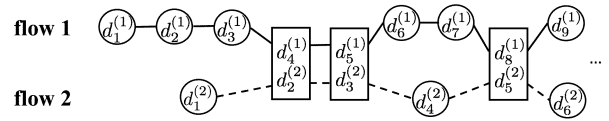}\vspace{-0.4cm}
\caption{Data items overlap between two flows}\label{fig:dataOverlap}
\vspace{-0.1cm}
\end{figure}
For example,  in Fig.~\ref{fig:dataOverlap}, we have $d^{(1)}_4\equiv d^{(2)}_2$, 
$d^{(1)}_5 \equiv d^{(2)}_3$ and $d^{(1)}_8 \equiv d^{(2)}_5$. 
Let $R_n$ denote the requested data item at time $\tau_n$.  
Thus, the event $\{I_n=k, R_n=d^{(k)}_i\}$ means that the request at time $\tau_n$ is from flow $k$ to fetch data item $d^{(k)}_i$.
We also abuse the notation for $R_n$ a bit and define $\Pr[R_0>x \mid I_0=k]$ to be the probability that the request at time $\tau_0$ is to fetch a data item with 
an index larger than $x$  in the ordered list $\left(d^{(k)}_i, i=1, 2, 3, \cdots\right)$ of flow $k$.  The ordering will be specified in the following part.

When the system reaches stationarity (Theorem~1 of~\cite{Jelenkovic99asymptoticapproximation}),  the miss ratio of the system is equal to the probability that a request $R_0$ at time $\tau_0=0$ 
finds that its asked data item is not kept in the cache.   
Therefore, we only need to consider $R_0$ in the following part. 
Due to multiple request flows, we have two sets of probabilities for each flow.  Flow~$k$ experiences the unconditional probabilities 
\begin{align}\label{eq:unconditional}
\Pr\left[R_0 = d^{(k)}_i\right]=p^{(k)}_i,  i=1, 2, 3, \cdots
\end{align}
and
the conditional probabilities 
\begin{align}\label{eq:conditional}
\Pr\left[R_0=d^{(k)}_i {\big |} I_0=k\right] = q^{(k)}_i, i=1, 2, 3, \cdots
\end{align}
In general,  $q^{(k)}_i$ can be very different from $p^{(k)}_i$, since the multiple
 request flows not only access distinct data items, but also share common data items, as shown in Fig.~\ref{fig:dataOverlap}.  
Let $\nu_k\eqdef \Pr[I_0=k]$.  We obtain, by~(\ref{eq:unconditional}), 
\begin{align}\label{eq:complex}
p^{(k)}_i = \sum_{j=1}^{M}\nu_j  \Pr\left[R_0=d^{(k)}_i {\big |} I_0=j\right]. 
\end{align}
Specially, if there is only a single flow $k$, i.e., $\Pr[I_0=k]=1$, 
 then $q^{(k)}_i=p^{(k)}_i$ for all $i$. 
 It couples the request flows, since a data item requested by flow $k$
 is more likely to be found in the cache when it has recently been requested by other flows.
 In this case, the usual belief is to pool these flows together, so that one flow can help the others to increase the hit.   
 However, if the fraction of overlapped data items is not significant enough, it is intuitively inevitable that the help obtained 
 from other flows on these common data items will be quite limited.  
 There have been no analytical studies to quantify the effects on how the overlapped data items can help different flows.

When studying flow $k$, assume that the data items $d^{(k)}_i$ are sorted such that the sequence $p^{(k)}_i$ is non-increasing with respect to $i$. 
Given~(\ref{eq:complex}),  the sequence $q^{(k)}_i$ is not necessarily non-increasing by this ordering.  
 We investigate how the following functional relationship $\Phi_k(\cdot)$ for flow~$k$, $1\leq k \leq M$, in a neighborhood of infinity, 
 impacts the miss ratio, 
 \begin{align}\label{eq:relation}
  \left(\sum_{i=y}^{\infty} q^{(k)}_i\right)^{-1}  \sim \Phi_k \left( \left(p^{(k)}_y\right)^{-1} \right), \; y \to \infty.
 \end{align}
Note $f(x)\sim g(x)$ means $\lim_{x\to \infty}f(x)/g(x)=1$.   The values in (\ref{eq:relation}) are defined using reciprocals, as both
$\left(\sum_{i=y}^{\infty} q^{(k)}_i\right)^{-1}$ and $\left(p^{(k)}_y\right)^{-1}$ take values in $[1, \infty)$, in line with the condition that $\Phi_k(\cdot)$ is defined in a neighborhood of infinity. 
We consider the following class of heavy-tailed distributions
 \begin{align}\label{eq:1class}
 \lim_{n\to \infty} q^{(k)}_n/q^{(k)}_{n+1}=1.
 \end{align}  
It includes Zipf's distribution $q^{(k)}_n\sim c/n^{\alpha}$, $c, \alpha>0$, and heavy-tailed Weibull distributions $q^{(k)}_n\sim d\exp\left(-c n^{\alpha}\right)$ with $c,d>0, 0<\alpha<1$.

It has been shown~\cite{fill:1996,Flajolet:1992,Jelenkovic99asymptoticapproximation,Jelenkovic:2004} 
that the miss probability of LRU is equivalent to the tail of the searching cost distribution under move-to-front (MTF).  
For MTF,  the data items are sorted in increasing order of their last access times.
Each time a request is made for a data item, this data item is moved to the first position of the 
list and all the other data items that were before this one increase their positions in the list by one.  
\begin{definition}\label{def:1}
Define $C_n$ to be the summation of the sizes for all the data items in the sorted list under MTF that
are  in front of the position of the data item requested by $R_n$ at time $\tau_n$.
\end{definition} 
If the cache capacity is $x$,  then a cache miss under MTF, which is equivalent for LRU policy,  can be denoted by $\{C_n>x\}$. 
For a special case when the data item sizes satisfy $s^{(k)}_i\equiv 1$ for all $k, i$, the event $\{C_n>x\}$ means the position of the data item in the list is larger than $x$ under MTF.

 For the $M$ flows mixed together,  let $\{d_i, i=1, 2, \cdots \}$ denote the set of data items requested by the entirety of these flows, with $\Pr[R_0=d_i]=p^{\circ}_i$.  
Let $s_i$ denote the size of data item $d_i$ and assume $\bar{s}\triangleq \sup_i s_i < \infty$. 
In general, $s_i$ can take different values when data item sizes vary. 
Let $$m(z)= \sum_{i=0}^{\infty} s_i\left( 1- \left(1-p^{\circ}_i\right)^z \right)$$ be an increasing function with an inverse $m^{\leftarrow}(z)$, which
is related to the Che approximation~\cite{che2002}.
 We can analytically derive $m^{\leftarrow}(z)$ in some typical cases, as shown in Corollaries~\ref{corollary:weibull} and \ref{corollary:3}, which directly exploit
 the properties of the popularity distributions, different from the Che approximation.  

 One of our main results can be informally stated as follows, for a gamma function $\Gamma(\beta_k+1)=\int_{0}^{\infty}y^{\beta_k}e^{-y}dy$.  
\vspace{0.1cm}

 {\bf Main Result (Intuitive Description)}
\emph {For $M$ flows sharing a cache, if $\Phi_k(x)$, $1 \leq k \leq M$, is approximately a polynomial function ($\approx x^{\beta_k}$), then, under mild conditions, 
we obtain, when the cache capacity $x$ is large enough,
\begin{align}\label{eq:missPinformal}
P[\mbox{miss ratio of flow } k] \approx \frac{\Gamma(\beta_k +1)}{\Phi_k(m^{\leftarrow}(x))}. 
\end{align}}
 
\noindent {\bf Sketch of the proof}:
First, we derive a representation for the miss probability of the request~$R_0$. 
Similar arguments have been used in~\cite{Fricker:2012,Jelenkovic99asymptoticapproximation} but we take a different approach.
 Among all the requests that occur before $\tau_0=0$ we find the last one that also requests data item~$R_0$.  
More formally, define $-\sigma$ to be the largest index of the request arrival before $\tau_0$ 
such that $R_{-\sigma} = R_0$.   
Conditional on $\{R_0=d^{(k)}_i\} \mathlarger{ \cap} \{I_0=k\}$, the following requests $R_{-1}, R_{-2}, R_{-3}, \cdots$ are i.i.d,  satisfying 
 \begin{align}
  \Pr \left[ R_{-j}=d^{(k)}_i \mathlarger{\mid} \{R_0=d^{(k)}_i\} \mathlarger{\cap} \{I_0=k\} \right]=p_i^{(k)}, j\geq 1,\nonumber
 \end{align}
 which implies
 \begin{align}
  \Pr \left[\sigma > n \mathlarger{\mid} \{R_0=d^{(k)}_i\} \cap \{I_0=k\} \right]=\left(1-p_i^{(k)} \right)^n. \nonumber
 \end{align}
Thus, unconditional on $R_0$, we obtain, recalling (\ref{eq:unconditional}) and (\ref{eq:conditional}),
 \begin{equation}\label{eq:rep_in}
  \Pr\left[\sigma>n \mathlarger{\mid} I_0=k\right]=\sum_{i=1}^{\infty} q_i^{(k)}\left(1-p_i^{(k)} \right)^n.
\end{equation}

Now we argue that the event $\{C_0>x\}$ is completely determined by the requests at the time points $\{ \tau_{-1}, \tau_{-2}, \cdots, \tau_{-\sigma} \}$. 
Let $M(n)$ denote the total size of all the distinct data items that have been requested on points  $\{ \tau_{-1}, \tau_{-2}, \cdots, \tau_{-n} \}$.  
Define the inverse function of $M(n)$ to be $M^{\leftarrow}(x)=\min\{n: M(n)\geq x \}$. 
We claim that 
\begin{align}\label{eq:representation_in}
  \{C_0>x\} = \{ \sigma >  M^{\leftarrow}(x)\}. 
\end{align}
If the event $\{ \sigma >  M^{\leftarrow}(x)\}$ happens,  the total size of the distinct data items requested on the time interval $(\tau_{-\sigma}, 0)$ is 
no smaller than $x$ and these data items are different from the one that is requested at 
time $\tau_0$ (or $\tau_{-\sigma}$). Due to the equivalence of LRU and MTF,  when $R_0$
arrives at $\tau_0$,  all of the data items requested on $(\tau_{-\sigma}, 0)$ will be listed in front of it under MTF.   Combining these two facts we obtain $\{ \sigma >  M^{\leftarrow}(x)\} \subseteq \{C_0>x\} $.
If $\{C_0>x\}$ occurs, then after $\tau_{-\sigma}$ when $R_0$ is listed in the first position of the list, there must be enough distinct data items that have been requested on $(\tau_{-\sigma}, 0)$ so that their total 
size exceeds or reaches $x$.  This yields $\{C_0>x\} \subseteq \{ \sigma >  M^{\leftarrow}(x)\}$, which proves (\ref{eq:representation_in}) and implies 
\begin{align}\label{eq:representation2_in}
  \Pr[C_0>x | I_0=k] = \Pr[\sigma >  M^{\leftarrow}(x) | I_0=k]. 
\end{align}
In order to compute $\Pr[\sigma >  M^{\leftarrow}(x) | I_0=k]$, we take two steps. The first step is to show
\begin{align}
\Pr[\sigma >  n | I_0=k] \approx \Gamma(\beta_k+1)/\Phi_k(n). \nonumber
\end{align}
 The second step 
is to relate $M^{\leftarrow}(x)$ to $m^{\leftarrow}(x)$ as $x \to \infty$.  

Here, we provide an intuitive proof for $\beta_k>0$.
From (\ref{eq:rep_in}), we have
 \begin{align} 
  \Pr \left[\sigma>n \mathlarger{\mid} I_0=k\right] & =  \sum_{i=1}^{\infty} q_i^{(k)} {(1-p_i^{(k)})}^n
   \approx \sum_{i=1}^{\infty} q_i^{(k)}e^{ -n p_i^{(k)}}, \nonumber
  \end{align}
which, in conjuction with (\ref{eq:relation}), yields,  by replacing $p_i^{(k)}$ by $1{\bigg /}\Phi_{k}^{\leftarrow}\left( \left(\sum_{j=i}^{\infty} q_j^{(k)}\right)^{-1} \right)$, 
\begin{align}\label{eq:firstStep_in}
 \sum_{i=1}^\infty q_i^{(k)}e^{ -n {\big /} \Phi_k \left(\left(\sum_{j=i}^{\infty} q_j^{(k)}\right)^{-1}\right)}    \approx \Gamma(\beta_k+1)/\Phi_k(n).
\end{align}

For the second step, we have $M(n) \approx m(n)$ with a high probability as $n \to \infty$ by a concentration inequality. 
The monotonicity and continuity of $m(n)$ imply $M^\leftarrow(x) \approx m^\leftarrow(x)$ with a high probability under certain conditions.
Applying (\ref{eq:representation2_in}) and (\ref{eq:firstStep_in}), we finish the proof
\begin{align}
   \Pr & [C_0>x | I_0=k]  = \Pr[\sigma >  M^{\leftarrow}(x) | I_0=k] \nonumber\\
   &  \approx \Pr[\sigma >  m^{\leftarrow}(x) | I_0=k]  \approx \Gamma(\beta_k+1)/\Phi_k\left(m^{\leftarrow}(x)\right). \nonumber
\end{align}

The rigorous proof is presented in Theorem~\ref{theorem:missP}.  It also provides a numerical method to approximate the miss probabilities. 
In practice, once we have the information about the data sizes $s_i$ and the corresponding data popularities $p_i^{\circ}$, e.g., from the trace, 
we can always explicitly express $m(z)$, since $i$ only takes a finite number of values in this case.  Then, we can
evaluate $m^{\leftarrow}(z)$ numerically; see Section~\ref{s:experiment}.   
Explicit expressions for $m^{\leftarrow}(z)$ are derived for some cases in Section~\ref{ss:mainresults}. 
Note that $m^{\leftarrow}(z)$ is tightly related to the Che approximation~\cite{che2002}; see Section~\ref{ss:che}.

\section{Multiple competing flows}\label{s:theory}
In this section, we rigorously characterize the miss probability
of a given request flow, say flow $k$,  when it is mixed with other competing flows that share the same cache in Section~\ref{ss:mainresults}. 
In Section~\ref{ss:decomposition}, we provide a method to calculate $m(x)$ for multiple flows based on a decomposition property.

\subsection{Asymptotic miss ratios} \label{ss:mainresults}
The miss probability of flow $k$, for a cache size $x$, is represented by a conditional probability $ \Pr[ C_0 > x   | I_0 = k]$.
  Recall $\bar{s}=\sup_{i} s_i<\infty$ and that $p^{\circ}_i=\Pr[R_0=d_i]$ is defined for the mixed flow.
  Note $m(z)\to \infty$ as $z \to \infty$.  By the theory of regularly varying functions~\cite{regularVariation}, a function $l(x): \RR^+ \to \RR^+ $ is slowly varying if for any $\lambda>0, l(\lambda x)/l(x) \to 1$
as $x \to \infty$;  and $\Phi(x)=x^{\alpha}l(x)$ is called regularly varying  of index $\alpha$. 

 Assume that, for a function $0<\delta(x)\leq 1$ and $\epsilon>0$,
  \begin{align}\label{eq:mz}
   \lim_{x\to \infty} \frac{m^{\leftarrow}\left( \left(1+ \epsilon \delta(x)\right) x \right)}{m^{\leftarrow}(x)}=f(\epsilon) \; \; \text{with}\;\; \lim_{\epsilon \to 0}f(\epsilon)=1. 
  \end{align}
 The function $\delta(x)$ characterizes how fast $m^{\leftarrow}(z)$ grows, and thus $\delta(x)$ should be selected to be as large as possible while still satisfying~(\ref{eq:mz}). 
  For example, when $m^{\leftarrow}(x)$ is regularly varying, e.g., $m^{\leftarrow}(x)=x^{\beta}$, we can let $\delta(x)=1$, which yields $f(\epsilon)=(1+\epsilon)^{\beta}$.
  When $m^{\leftarrow}(x)=e^{x^{\xi}}, 0<\xi<1$,  we can pick $\delta(x)=x^{-\xi}$, since $\lim_{x\to \infty}e^{ \left(x + \epsilon x^{1-\xi}\right)^{\xi} }/e^{x^{\xi}} = e^{\epsilon \xi}$, implying
  $f(\epsilon)=e^{\epsilon \xi}$. Both satisfy $\lim_{\epsilon \to 0}f(\epsilon)=1$.
  Note that in these examples $\delta(x)$ satisfies the following condition: there exist $h_2> h_1>0, h_4>h_3>0$ and $x_0$, for $x>x_0$,  
  \begin{align}\label{eq:delta}
     h_1<\frac{\delta(x)}{\delta(x+\epsilon \delta(x))}<h_2,    h_3<\frac{\delta(x-\epsilon \delta(x))}{\delta(x)}<h_4.
  \end{align}

  \begin{theorem}\label{theorem:missP}
Consider $M$ flows sharing a cache. Under~(\ref{eq:1class}), (\ref{eq:mz}) and~(\ref{eq:delta}),  for  $\Phi_k(x)\sim x^{\beta_k} l_k(x)$, $1 \leq k \leq M$, and
$\varlimsup_{x\to \infty} \log \left(m^{\leftarrow}(x)\right)/(\delta^2(x) x) =0$, 
we have, for $\beta_k \geq 0$ (when $\beta_k=0$, $l_k(x)$ is eventually non-decreasing), as $x\to \infty$,
\begin{align}\label{eq:missP}
\Pr[C_0 >x | I_0=k] \sim \frac{\Gamma(\beta_k +1)}{\Phi_k(m^{\leftarrow}(x))}.
\end{align}
\end{theorem}
Theorem~\ref{theorem:missP} is the rigorous version of the main result described in (\ref{eq:missPinformal}). The proof is presented in Section~\ref{ss:p:theorem:missP}.
Based on Theorem~\ref{theorem:missP}, we can easily derive some corollaries.  
We begin with the {special case} when there is only a single flow~$k$ in service and all data items are of the same size $s_i\equiv 1$. 
For a single flow~$k$, we simplify the notation by $ \Pr\left[R_0>x {\big |} I_0=k \right] $ $=  \Pr\left[R_0>x \right]$ and  $ \Pr\left[C_0>x {\big |} I_0=k \right] =  \Pr\left[C_0>x \right]$.
 Theorem~\ref{theorem:missP} recovers the results 
in~\cite{Jelenkovic99asymptoticapproximation,Jelenkovic:2004} for Zipf's distribution
\begin{align}\label{eq:zipf}
p^{(k)}_i =  q^{(k)}_i \sim c/i^{\alpha}, \alpha>1.
\end{align}
Our result enhances (\ref{eq:zipf}) in three aspects. First, we study multiple flows ($p^{(k)}_i \neq q^{(k)}_i$) that can have overlapped data items and the requested data items can have different sizes. 
 Second,  we address the case $\alpha= 1$ (then $c$ needs to be replaced by $l(i)$ as in~(\ref{eq:zipf5})), 
 while the 
results in~\cite{Jelenkovic99asymptoticapproximation,Jelenkovic:2004} assume $\alpha>1$.  For $\alpha<1$, we need to assume that only a finite number of data items can be requested (this paper assumes an infinite number); otherwise the popularity distribution does not exist. This special case needs to 
be handled differently and is not presented in this paper. Due to this difference, the asymptotical result in (\ref{eq:missP}) is only accurate for 
large $x$ when $\alpha \approx 1$. 
Third, our result
can derive the asymptotic miss probability for a large class of popularity distributions, e.g., Weibull, with varying data item sizes.  
Corollary~\ref{corollary:zipf2} extends the results of Theorem~3 in~\cite{Jelenkovic99asymptoticapproximation} that is proved under the condition~(\ref{eq:zipf}) to regularly varying probabilities
\begin{align}\label{eq:zipf5}
p^{(k)}_i \sim l(i)/i^{\alpha}, \alpha \geq 1,
\end{align} 
with $l(\cdot)$ being a slowly varying function, e.g., $l(x)=\log x$.

\begin{corollary} \label{corollary:zipf2}
Consider a single flow with $s_i\equiv 1$ and $p_i^{(1)} \sim l(i)/i^\alpha$, $\alpha>1$.
Let $l_1(x) = {l(x)}^{-1/\alpha}$ and $l_{n+1} = l_1(x/l_{n}(x))$, $n\geq 1$.
If $l_{n_0}(x) \sim l_{n_0+1}(x)$ as $x \to \infty$ for some $n_0$, then
\begin{align}\label{eq:zipf4}
 \lim_{x\to \infty}\frac{\Pr[C_0 >x]}{\Pr[R_0>x]} = (1-1/\alpha) \left(\Gamma(1-1/\alpha)\right)^{\alpha}.
\end{align}
\end{corollary}

\begin{proof}
First we provide a proof for the special case $l(x)=c$, which was proved in Theorem~3 of~\cite{Jelenkovic99asymptoticapproximation}.
The proof for a general $l(x)$ is presented in Section~\ref{ss:p:corollary:zipf2}.

Note that $p_x^{(k)} \sim c/x^{\alpha}$ and
\begin{align}\label{eq:zipf2}
  \Pr\left[R_0>x\right] = \sum_{i\geq x}p_i^{(1)} \sim \int_{x}^{\infty} \frac{c}{x^{\alpha}}dx = \frac{c}{(\alpha-1)x^{\alpha-1}}.
\end{align}
Using~(\ref{eq:relation}), we obtain $\Phi_1(x)\sim (\alpha-1)c^{-1/\alpha} x^{1-1/\alpha}$. 
 In addition, we have 
 \begin{align}
 m(z) &\sim \sum_{i\geq 1} \left(1-\exp\left( -\frac{cz}{i^{\alpha}} \right)\right) 
      \sim \int_{1}^{\infty} \left(1-\exp\left( -\frac{cz}{x^{\alpha}} \right)\right) dx \nonumber\\
      &\sim \Gamma(1-1/\alpha) c^{1/\alpha} z^{1/\alpha},  \nonumber
 \end{align}
 implying the inverse
$m^{\leftarrow}(z) \sim z^{\alpha}/\left(c \Gamma(1-1/\alpha)^{\alpha} \right).$
Picking $\delta(x)=1$, it is easy to verify $\varlimsup_{x\to \infty}  \log m^{\leftarrow}(x)/\delta(x)^2 x =0$
and~(\ref{eq:delta}).  
Therefore, by Theorem~\ref{theorem:missP}, we obtain, as $x\to \infty$, 
\begin{align}\label{eq:zipf1}
\Pr[C_0>x] \sim \frac{\Gamma(2-1/\alpha) \Gamma(1-1/\alpha)^{\alpha-1}}{\alpha-1}\frac{c}{x^{\alpha-1}}.
\end{align}   
Combining~(\ref{eq:zipf1}) and (\ref{eq:zipf2}) finishes the proof. 
\end{proof}

\begin{corollary}\label{corollary:weibull}
For a single flow with requests following a heavy-tailed Weibull distribution $p^{(k)}_i \sim c\exp\left(- i^{\xi} \right)$, $0<\xi<1/3$ and $s_i\equiv 1$, 
we have, for a Euler's constant $\gamma=0.5772\cdots$, 
\begin{align}\label{eq:weibull}
 \lim_{x\to \infty}\frac{\Pr[C_0 >x]}{\Pr[R_0>x]} = e^{\gamma}.
\end{align}
\end{corollary}
\begin{proof}
Since $c e^{- x^{\xi}}$ is a decreasing function in $x$, we have
\begin{align}\label{eq:weibull5}
 \int_{x}^{\infty} ce^{- y^{\xi}}dy \leq  \sum_{i=x}^{\infty}   c\exp\left(- i^{\xi} \right) & \leq \int_{x-1}^{\infty} ce^{- y^{\xi}}dy. 
\end{align}
Changing the variable $z=y^{\xi}$ and using the property of incomplete gamma function, we obtain
\begin{align}\label{eq:weibull6}
 \int_{x}^{\infty} ce^{- y^{\xi}}dy  &= \int_{x^{\xi}}^{\infty} \frac{c}{\xi}z^{1/\xi -1}e^{-z}dz \sim \frac{c}{\xi}x^{1-\xi}e^{-x^{\xi}},
\end{align}
which implies, for $0<\xi<1$,
\begin{align}\label{eq:weibull1}
\Phi_k(x) \sim \xi \left( \log (c x) \right)^{1-1/\xi} x.
\end{align} 
Using Lemma~6 in~\cite{Jelenkovic99asymptoticapproximation}, we obtain 
\begin{align}\label{eq:weibull2}
m^{\leftarrow}(z) \sim e^{-\gamma} e^{z^{\xi}}/c.
\end{align} 
Picking $\delta(x)=x^{-\xi} >0$, for $0<\xi<1/3$, it is easy to verify  $\varlimsup_{x\to \infty}$ $ \log m^{\leftarrow}(x)/x^{1-2\xi} =0$
and~(\ref{eq:delta}). 
Combining~(\ref{eq:weibull1}) and (\ref{eq:weibull2}), by Theorem~\ref{theorem:missP}, we derive
\begin{align}\label{eq:weibull3}
\Pr[C_0>x] \sim \frac{e^{\gamma}c}{\xi}x^{1-\xi}e^{-x^{\xi}},
\end{align} 
which, using~(\ref{eq:weibull5}) and (\ref{eq:weibull6}),  proves (\ref{eq:weibull}).
\end{proof}

 \subsection{Decomposition property} \label{ss:decomposition}
For multiple request flows without overlapped common data items,  we have a decomposition property. 
Let $P=\left(p^{\circ}_i, i\geq 1\right)$ be constructed from a set of distributions $Q^{(k)}=\left(q^{(k)}_i, i\geq 1\right)$ 
according to probabilities $\nu_k$, $\sum_{k}\nu_k=1$.
Specifically, a random data item following the distribution $P$ is generated by sampling from the distribution $Q^{(k)}$ with a probability $\nu_k$.  
Since two flows~$k_1,k_2$ have no overlapped data items, we have 
$\Pr\left[R_0=d^{(k_1)}_i {\big |} I_0=k_2\right]=0$. 
Therefore, according to~(\ref{eq:complex}), $\left(p^{\circ}_i, i\geq 1\right)$ can be represented by an unordered list, 
\begin{align}\label{eq:unordered}
\left( \left( \nu_k q^{(k)}_i, k+i=m\right), m=2, 3, 4, \cdots \right).
\end{align} 
Let $\bar{m}(z)= \sum_{i=0}^{\infty} s_i\left( 1- \exp(-p^{\circ}_i z) \right)$. 
Lemma~\ref{lemma:decomposition} shows a decomposition property for $m(z)$ 
and $\bar{m}(z) \sim m(z)$ under certain conditions.
Let $m^{(k)}(z) = \sum_{i=0}^{\infty} s^{(k)}_i\left( 1- \left(1-q^{(k)}_i\right)^z \right)$ 
and $\bar{m}^{(k)}(z) = \sum_{i=0}^{\infty} s^{(k)}_i\left( 1- \exp\left(-q^{(k)}_i z\right) \right)$.
It is often easier to  compute $\bar{m}^{(k)}(z)$ than $m^{(k)}(z)$.  

\begin{lemma}\label{lemma:decomposition}
Without overlapped data items,  if,  for either $g(x)=m^{(k)}(x)$ or $g(x)=\bar{m}^{(k)}(x)$, we have 
$\varlimsup_{x\to \infty}g((1+\delta) x)/g(x)=f^{(k)}(\delta)$, $0<\delta<1$ with $\lim_{\delta \to 0} f^{(k)}(\delta)=1$,  then,  as $z \to \infty$,
\begin{align}\label{eq:approx}
  m^{(k)}(z) \sim \bar{m}^{(k)}(z), 
\end{align}
and
\begin{align}\label{eq:decomp}
 \bar{m}(z)\sim  \sum_{k } \bar{m}^{(k)}(\nu_kz)  \sim \sum_{j} m^{(k)}(\nu_kz) \sim m(z).
\end{align}
\end{lemma}

The proof of Lemma~\ref{lemma:decomposition} is presented in Section~\ref{s:proof}.
It can be used to compute $m(x)$ for multiple flows sharing the same cache. 
Furthermore, applying Theorem~\ref{theorem:missP}, we can derive the miss probability for each flow. 

\begin{corollary}\label{corollary:3}
Consider $M$ flows without overlapped data, satisfying $\Pr[ R_0 =d^{(k)}_x | I_0 =k ] \sim c_k/x^{\alpha_k}$, $1\leq k \leq M$
and $\Pr[ I_0 =k]=\nu_k$, $\sum_{k=1}^M \nu_k=1$.  
Assume that the data items of flow~$k$ have identical sizes, i.e. $s_i^{(k)} = s^{(k)}, i \geq 1$.
For $\widetilde{\alpha}_1 \triangleq \min_{1\leq k \leq n} \alpha_k$ and $S_1=\{k \in {\mathbb Z} | \alpha_k = \widetilde{\alpha}_1, 1 \leq k \leq M\}$,
we have, 
for $k \in S_1$,
 \begin{align}
 \Pr[C_0>x | I_0 =k] \sim \frac{\Gamma(2-1/\widetilde{\alpha}_1)}{\widetilde{\alpha}_1-1}  \frac{{\gamma_1}^{\widetilde{\alpha}_1-1}} {{(\nu_kc_k)}^{1-\frac{1}{\widetilde{\alpha}_1}}} \frac{c_k}{x^{\widetilde{\alpha}_1-1}}, \nonumber
 \end{align} 
and for $k \in {S_1}^c$,
 \begin{align}
 \Pr[ C_0>x | I_0 =k] \sim \frac{\Gamma\left(2-1/\alpha_k \right) } {\alpha_k-1}   \frac{{\gamma_1}^{\widetilde{\alpha}_1-\frac{\widetilde{\alpha}_1}{\alpha_k}}} {{(\nu_kc_k)}^{1-\frac{1}{\alpha_k}}}
         \frac{c_k}{x^{\widetilde{\alpha}_1 - \frac{\widetilde{\alpha}_1}{\alpha_k}}}, \nonumber
 \end{align}
 where
\begin{align}
\gamma_1=\Gamma(1-1/\widetilde{\alpha}_1) \sum_{k \in S_1} s^{(k)}  (c_k \nu_k )^{1/\widetilde{\alpha}_1}. \label{eq:gamma1}
\end{align}
\end{corollary}

\begin{proof}
For flow~$k$, $1 \leq k \leq n$, we have
 $$\Phi_k (x)\sim (\alpha_k-1)c_1^{-1/\alpha_k}(\nu_k x)^{1-1/\alpha_k},$$
 $$m^{(k)}(z) \sim s^{(k)} \Gamma(1-1/\alpha_k)  (c_k z)^{1/\alpha_k} .$$ 
Using the decomposition property of Lemma~\ref{lemma:decomposition},  
 $m(z)$ is asymptotically determined by flows with indices in $S_1$, 
  $$m(z) \sim \sum_{k \in S_1} s^{(k)}  (c_k \nu_k )^{1/\widetilde{\alpha}_1} \Gamma(1-1/\widetilde{\alpha}_1) z^{1/\widetilde{\alpha}_1},$$ 
 implying that $m^{\leftarrow}(z)$ is asymptotically equal to
 {\small
 \begin{align} \label{eq:minverseC3}
  z^{\widetilde{\alpha}_1} \left/ {\left( \sum_{k \in S_1} s^{(k)}  (c_k \nu_k )^{1/\widetilde{\alpha}_1} \Gamma(1-1/\widetilde{\alpha}_1) \right)}^{\widetilde{\alpha}_1} \right..
 \end{align}
 }
 \\
 \noindent Now, by Theorem~\ref{theorem:missP},  we can prove the corollary after straightforward computations. 
\end{proof}

Corollary~\ref{corollary:3} approximates the miss probabilities for multiple flows with different $\alpha_k$ when the cache capacity $x \to \infty$. 
When the cache capacity is small, this approximation is not accurate. 
In order to improve the accuracy,  denote by $\widetilde{\alpha}_2\triangleq \min_{k \in {S_1}^c} \alpha_k$ the second smallest value among all $\alpha_k$'s.
Defining $S_2=\{k \in {\mathbb Z} | \alpha_k = \widetilde{\alpha}_2, 1 \leq k \leq n\}$, 
we consider all flows in the set $\{k: k \in S_1 \cup S_2\}$, and derive
\begin{align}
m(z) \sim & \sum_{k \in S_1} s_k  (c_k \nu_k )^{1/\alpha_1^*} \Gamma(1-1/\alpha_1^*) z^{1/\alpha_1^*}  \nonumber\\
 & + \sum_{k \in S_2} s_k  (c_k \nu_k )^{1/\alpha_2^*} \Gamma(1-1/\alpha_2^*) z^{1/\alpha_2^*}. \nonumber
\end{align}
The inverse function of $m(z)$ can be better approximated by
\begin{align}
 m^\leftarrow(z) \sim z^{\widetilde{\alpha}_1} \left/ \left( \gamma_1+\gamma_2 (z/\gamma_1)^{\widetilde{\alpha}_1/\widetilde{\alpha}_2-1} \right)^{\widetilde{\alpha}_1} \right. , \label{eq:minverseaprxC3}
\end{align}
where $\gamma_2= \Gamma(1-1/\widetilde{\alpha}_2) \sum_{k \in S_2} s^{(k)}  (c_k \nu_k )^{1/\widetilde{\alpha}_2}$ and $\gamma_1$ is defined in (\ref{eq:gamma1}).
We obtain more accurate numerical results for miss probabilities using (\ref{eq:minverseaprxC3}) instead of (\ref{eq:minverseC3}) 
especially when the cache capacity is small, though the expressions in (\ref{eq:minverseC3})  and (\ref{eq:minverseaprxC3}) are asymptotically equivalent. 
Experiments~\ref{exp:2} in Section~\ref{s:experiment} validates this approximation. 
Alternatively, we also resort to numerical methods to directly evaluate $m^\leftarrow(z)$ for more complex cases.

\subsection{Connection to the Che approximation}\label{ss:che}
The miss probability of LRU algorithm has been extensively studied 
using the Che approximation~\cite{che2002}.  Now we show that 
the Che approximation is asymptotically accurate under certain conditions; see a related validity argument in~\cite{leaseMiss}.  
For multiple flows, 
the overall miss probability 
computed by the Che approximation is
\begin{align}
 \Pr_{che}[C_0>x|I_0=k] = \sum_{i=1}^\infty q_i^{(k)} e^{-p_i^{(k)}T}, \nonumber
\end{align}
where $T$ is the cache characteristic time as the unique solution to $
 \sum_{i=1}^\infty s_i ( 1- e^{-p_i^\circ T})=x$. 
\begin{theorem} \label{lemma:che}
 Under the conditions of Theorem~\ref{theorem:missP},  we have, as $x\to \infty$,
\begin{align}
\Pr_{che}[C_0>x|I_0=k] \sim \Pr[C_0>x|I_0=k] \sim \frac{\Gamma(1+\beta_k)}{\Phi_k(m^\leftarrow(x))}. \nonumber
\end{align}
\end{theorem}
The proof of Theorem~\ref{lemma:che} is presented in Section~\ref{ss:p:lemma:che}.

\section{Pooling and separation} \label{s:PorS} 
We first characterize the self-organizing behavior of LRU caching for multiple flows in Section~\ref{ss:noOverlap}.  
Then, we study how the interactions of competing flows impact the individual ones in Section~\ref{ss:individualFlow}. 
The consequences of overlapped data items across different flows are investigated in Section~\ref{ss:overlap}.
Based on the insights,  we discuss engineering implications in Section~\ref{ss:engrImplication}.

A pooling scheme serves the $M$ request flows jointly using the cache space of size $x$.
A separation scheme divides the cache space $x$ into $M$ parts according to fractions $\{u_k\}_{1\leq k \leq M}$, $\sum_{k=1}^M u_k =1$, and allocates $u_kx$ to flow~$k$. 

\subsection{Self-organizing behavior of pooling}\label{ss:noOverlap}

Based on the asymptotic miss ratios derived in Theorem~\ref{theorem:missP}, we show that, when multiple flows have similar distributions and identical data item sizes,
resource pooling asymptotically gives the best overall hit ratio achieved by the optimal separation scheme.    Otherwise, the optimal separation scheme results in a better 
overall miss ratio. Note that the optimal separation scheme is static while the pooling scheme is adaptive without any parameter tuning or optimization. This explains why pooling is better in Fig.~\ref{fig:requestRate}. 
Denote by $\Pr_{s}^*[C_0>x] $ and $\Pr_{p}[C_0>x]$ the overall miss probabilities 
under the optimal separation $\{u^*_k\}$ and under resource pooling, respectively.

\begin{theorem} \label{lemma:samealpha}
For $M$ flows without overlapped data, following $\Pr\left[R_0=d_x^{(k)} | I_0=k \right] \sim c_k/x^{\alpha_k}$, $\alpha_k>1$, $1\leq k \leq M$ 
and the data items of flow~$k$ having the same size $s_i^{(k)}=s^{(k)}$, we have
\begin{align}\label{eq:adaptive}
  \lim_{x\to \infty}  \Pr_{p}[C_0>x]/ \Pr_{s}^*[C_0>x] \geq 1,
\end{align}
and the equality holds  if and only if $s^{(1)}=s^{(2)}=\cdots=s^{(M)}$.
\end{theorem}
This result explains the simulation in Fig.\ref{fig:diffSize} when data item sizes are different.  In practice, data item sizes 
vary, and they can be considered approximately equal if within the same range, as used by slabs of Memcached~\cite{JiangSong,memcachNSDI}. 
Note that Theorem~\ref{lemma:samealpha} only characterizes an asymptotic result.  When the cache size is not large enough and $\alpha_k$'s are different, 
resource pooling can be worse than the optimal separation, as studied in~\cite{chu2016allocating}.  As commented after Corollary~\ref{corollary:3}, a better 
approximation for small cache sizes is to use Theorem~\ref{theorem:missP} by numerically evaluating $m^{\leftarrow}(x)$.  Theorem~\ref{lemma:samealpha} also shows that when data item sizes vary significantly,  resource pooling could be worse than separation, as illustrated in Fig.~\ref{fig:diffSize}. 
\begin{proof} 
First, we assume $\alpha_k=\alpha, 1\leq k \leq M$. 
To characterize resource separation,  by Theorem~\ref{theorem:missP}, we obtain
\begin{align}
 \Pr_{s}[C_0>x] & = \sum_{k=1}^M \Pr[I_0=k] \Pr_{s}[C_0>u_k x | I_0=k] \nonumber \\
 & \sim \sum_{k=1}^M \nu_k c_k \frac{{\Gamma(1-1/\alpha)}^\alpha}{\alpha} {\left( \frac{s^{(k)}} {u_k x} \right)}^{\alpha-1}. \label{eq:pstatic}
\end{align}
Since the optimal separation method $u^*=(u_1^*,u_2^*,\cdots,u_k^*)$ minimizes the overall asymptotic miss probability, we have
\begin{align}
 \text{Minimize} \qquad & \sum_{k=1}^M \nu_k c_k \frac{{\Gamma(1-1/\alpha)}^\alpha}{\alpha} {\left( \frac{s^{(k)}} {u_k x} \right)}^{\alpha-1} \nonumber \\
 \;\text{Such that} \qquad & \sum_{k=1}^M u_k = 1,    u_k \geq 0. \nonumber 
\end{align}
The solution of this convex optimization problem satisfies the KKT conditions, and therefore, 
\begin{align}
 \frac{c_1 \nu_1 {(s^{(1)})}^{\alpha-1}}{{u_1^*}^\alpha} = \frac{c_2 \nu_2 {(s^{(2)})}^{\alpha-1} }{{u_2^*}^\alpha} = \cdots =\frac{c_M \nu_M {(s^{(M)})}^{\alpha-1}}{{u_M^*}^\alpha}, \nonumber 
\end{align}
resulting in
\begin{align}
 u_k^*= \frac{{(c_k \nu_k)}^{1/\alpha} {(s^{(k)})}^{1-1/\alpha}}{\sum_{i=1}^M {(c_i\nu_i)}^{1/\alpha}  {(s^{(i)})}^{1-1/\alpha}} ,  \quad k=1,2,\cdots,M.  \label{eq:solution1}
\end{align}
 From (\ref{eq:pstatic}) and (\ref{eq:solution1}), we obtain 
 {\small
\begin{align}
 \Pr_{s}^*[C_0>x]
 &\sim \frac{{\Gamma(1-1/\alpha)}^\alpha}{\alpha x^{\alpha-1}} {\left( \sum_{k=1}^M {(c_k\nu_k)}^{1/\alpha}{(s^{(k)})}^{1-1/\alpha}\right)}^\alpha . \label{eq:pdivide}
\end{align}
}
To study resource pooling, 
 we obtain, by Corollary~\ref{corollary:3},
\begin{align} 
 & \Pr_{p}  [C_0>x]  = \sum_{k=1}^M \Pr[I_0=k] \Pr_{p}[C_0>x | I_0=k] \nonumber \\
 & \sim \sum_{k=1}^M \nu_k \frac{{\Gamma(1-1/\alpha)}^\alpha}{\alpha} \frac{{\left(\sum_{i=1}^M {(c_i\nu_i)}^{1/\alpha} s^{(i)}\right)}^{\alpha-1}} {{\nu_k}^{1-1/\alpha}}  \frac{{c_k}^{1/\alpha}}{x^{\alpha-1}} \nonumber \\
 & =  \frac{{\Gamma(1-1/\alpha)}^\alpha}{\alpha x^{\alpha-1}} \left( \sum_{k=1}^M {( c_k \nu_k )}^{1/\alpha} \right)   {\left(\sum_{i=1}^M {(c_i \nu_i)}^{1/\alpha} s^{(i)}\right)}^{\alpha-1}, \nonumber
\end{align}
which, using (\ref{eq:pdivide}) and H\"{o}lder's inequality, proves (\ref{eq:adaptive}). 
The equality holds if and only if $s^{(1)}=s^{(2)}=\cdots = s^{(n)}$.

Now,  if $\alpha_k$'s are not identical,  let $\widetilde{\alpha}_1=\min_{1 \leq k \leq n} \alpha_k$ and $S_1=\{k \in {\mathbb Z} | \alpha_k = \widetilde{\alpha}_1, 1 \leq k \leq n\}$.
By Corollary~\ref{corollary:3}, we have, for resource pooling, 
\begin{align}
\Pr_{p}[C_0>x]  
 \sim & \sum_{k \in S_1} \Pr[I_0=k] \Pr_{p}[C_0>x | I_0=k]. \nonumber 
\end{align}

For separation,  by (\ref{eq:zipf4}), we have, as $x \to \infty$,
\begin{align}
 \Pr_{s}[C_0>u_kx | I_0=k] \sim \frac{{\Gamma(1-1/\alpha_k)}^{\alpha_k}}{\alpha_k} \frac{c_k}{{(u_kx)}^{\alpha_k-1}}. \nonumber
\end{align}
Thus, the overall miss probability is
\begin{align}
 \Pr_{s}[C_0>x] 
 & \sim  \sum_{k \in S_1} \nu_k c_k \frac{{\Gamma(1-1/\alpha_k)}^{\alpha_k}}{\alpha_k} {\left(\frac{s^{(k)}}{u_kx}\right)}^{\alpha_k-1}. \nonumber
\end{align}
Thus, the same arguments for the case $\alpha_k=\alpha$ can be repeated to prove (\ref{eq:adaptive}) in this case. 
\end{proof}

\subsection{Impacts on individual flows}\label{ss:individualFlow}
When the QoS (quality-of-service) of individual flows is important,  we need to guarantee the miss ratio of each flow.  
The following theorem shows that,  for each flow,  cache pooling 
asymptotically achieves the same miss ratio as the optimal separation under certain conditions. 
Interestingly, the miss ratios of multiple competing flows decrease according to $c_k^{1/\alpha}\nu_k^{1/\alpha-1}, 1\leq k \leq M$
when sharing the same cache. 
\begin{corollary}\label{corollary:4}
For $M$ flows under the conditions of Theorem~\ref{lemma:samealpha} with $s^{(1)}=s^{(2)}=\cdots=s^{(M)}$ and $\alpha_k=\alpha$, 
we have
\begin{align}
  \lim_{x\to \infty}  \Pr_{p}[C_0>x | I_0=k]/ \Pr_{s}^*[C_0> u^{\ast}_kx | I_0=k] = 1. \nonumber
\end{align}
Furthermore, the miss ratios of any two flows $i, j$ satisfy  
\begin{align}
\lim_{x\to \infty} & \frac{\Pr_{p}[C_0>x | I_0=i]}{ \Pr_{p}[C_0>x | I_0=j]} = 
\lim_{x\to \infty} \frac{ \Pr_{s}^*[C_0> u^{\ast}_i x | I_0=i]}{\Pr_{s}^*[C_0> u^{\ast}_j x | I_0=j]} \nonumber\\
&=\frac{c_i^{1/\alpha}\nu_i^{1/\alpha-1} }{ c_j^{1/\alpha}\nu_j^{1/\alpha-1} }. \nonumber
\end{align} 
\end{corollary}
The proof of this corollary is based on the same arguments in the proof of Theorem~\ref{lemma:samealpha}. 
 This result quantifies the empirical observation~\cite{JiangSong} 
 that mixing multiple flows benefits the ones with large arrival rates at the expense of the others 
 with small arrival rates. It also shows that the popularity distributions need to be considered if $c_i\neq c_j$.  
Therefore,  if arrival rates differ significantly,  mixing flows requires extra caution. 
Simulations for validating Corollary~\ref{corollary:4} is in Section~\ref{s:experiment}.

\subsection{Overlapped data items}\label{ss:overlap}
In this section, we show that when overlapped data items exceed certain levels, pooling cache 
space together can even improve the performance of every flow. Since overlapped data items across more than two flows are complicated, 
we only consider two flows with unit-sized data items. 
Notably,  there always exists
a good region of parameters such that the miss probabilities of both flows under pooling are better than under separation; see Experiment~\ref{exp:3} in Section~\ref{s:experiment}.  
    
Since the requested data items can overlap (see Fig.~\ref{fig:dataOverlap}),
we introduce $3$ disjoint classes of data items, $A, B$ and $D$ for the two flows. 
Flow~1 and~2 request data items from class $A=\{d^{(A)}_i, i=1, 2, \cdots\}$
and class $B=\{d^{(B)}_i, i=1, 2, \cdots\}$, respectively. 
Class $D=\{d^{(D)}_i, i=1, 2, \cdots\}$ represents the common data items that are requested by both flow~1 and~2. 
We use $J_n=A, J_n=B, J_n=D$ to indicate that the request~$n$ is for class~$A$, $B$ and $D$,  respectively.  
Let $\Pr \left[ J_n=A  \mid I_n=1 \right]=p^{(1)}_A$, $\Pr \left[ J_n=D  \mid I_n=1 \right]=p^{(1)}_D$
with $p^{(1)}_A+p^{(1)}_D=1$, and 
$\Pr \left[ J_n=B  \mid I_n=2 \right]=p^{(2)}_B$, $\Pr \left[ J_n=D  \mid I_n=2 \right]=p^{(2)}_D$ with $p^{(2)}_B+p^{(2)}_D=1$.   
Class~$A$ and~$B$ have $\Pr\left[ R_0 =d^{(A)}_x  \mathlarger{|} I_0 =1, J_0=A \right]$ $\sim c_A/x^{\alpha}$,   
$\Pr\left[ R_0 =\right. $ $\left. d^{(B)}_x | I_0 =2, J_0=B \right]$ $\sim c_B/x^{\alpha}$.
For class~$D$,  we assume that
$\Pr\left[ R_0 =d^{(D)}_x | I_0 =k, J_0=D \right] \sim c_{D}/x^{\alpha}, k=1, 2$. 

An optimal separation scheme has been proposed 
in~\cite{dehghan2017sharing} to serve classes $A,B,D$ in three isolated parts of the whole cache space.
Since the three classes do not have overlapped data items, Theorem~\ref{lemma:samealpha} implies that
the optimal separation is asymptotically equivalent to pooling. 
However, this isolation scheme requires a lot of tracking information, and is difficult to implement in practice.  
We consider a practical constraint that 
a flow is the smallest unit that cannot be further divided into sub-flows.   
In this case, the optimal separation is not always the best. In fact, it can be worse than resource pooling if enough data overlap is present.

For a static separation $u=(u_1,u_2)$, define a good region $\mathcal{G}_u$
for positive parameters $P=\left( \nu_1,\nu_2, c_A, c_B, p_A^{(1)}, p_B^{(2)}, p_D^{(1)},\right.$ $\left.p_D^{(2)} \right)$, which satisfy, 
for $p_D^* = p_D^{(1)}\nu_1+p_D^{(2)}\nu_2 >0$,
{\small
\begin{align}\label{eq:assumption1}
 \frac{{\left( \left(c_Ap_A^{(1)} \right)^{1/\alpha} + \left(c_Dp_D^{(1)}\right)^{1/\alpha} \right)}^\alpha} 
         {{\left( \left(c_Ap_A^{(1)}\nu_1\right)^{1/\alpha} + \left(c_Bp_B^{(2)}\nu_2\right)^{1/\alpha} + \left(c_Dp_D^*\right)^{1/\alpha} \right)}^{\alpha-1}} \nonumber \\
 > {u_1}^{\alpha-1} \left( \frac{\left(c_Ap_A^{(1)}\right)^{1/\alpha}}{{\nu_1}^{1-1/\alpha}} 
 + \frac{{c_D}^{1/\alpha}p_D^{(1)}}{{p_D^*}^{1-1/\alpha}} \right),
\end{align}
}
\\
\noindent with another symmetric constraint that replaces  $u_1$, $c_A$, $c_B$, $p_A^{(1)}$, $p_B^{(2)}$, $p_D^{(1)}$ in (\ref{eq:assumption1}) with $u_2$, $c_B$, $c_A$, $p_B^{(2)}$, $p_A^{(1)}$, $p_D^{(2)}$, respectively. 
The following corollary shows that when the parameters satisfy $P \in \mathcal{G}_u$, 
both flows have smaller miss ratios by resource pooling than by the static separation $u$.
Remarkably,  the parameters in $P$ for the optimal static separation $u^{\ast}$ that minimizes the overall miss ratio (defined in Section~\ref{ss:noOverlap}) are
always in the good region $\mathcal{G}_{u^{\ast}}$, although this region is defined to study the miss ratios of individual flows. 
\begin{corollary}\label{corollary:5}
For any positive $u = (u_1,u_2)$, if $P \in \mathcal{G}_u$ and is strictly positive, then we have, for $k=1, 2$,
\begin{align}
\lim_{x \to \infty }\Pr_p[C_0>x|I_0=k] / \Pr_s[C_0>u_kx|I_0=k] < 1. \nonumber 
\end{align}
Furthermore, if $u=u^{\ast}$, then we always have $P \in \mathcal{G}_{u^{\ast}}$.
\end{corollary}
The proof is a straightforward computation based on Theorem~\ref{theorem:missP}. 
This corollary also implies that $\mathcal{G}_u$ is always nonempty.
We use simulations to validate Corollary~\ref{corollary:5} in Section~\ref{s:experiment}.

\subsection{Engineering Implications} \label{ss:engrImplication}
Whether resource pooling or separation should be used for LRU caching is complicated. There are no straight yes or no answers, 
depending on four critical factors. They include the popularity distributions, request rates, data item sizes and overlapped 
data across different flows.  This problem becomes even more complicated due to engineering issues. However,  there are still guidelines to improve the miss ratios.  

Our analysis shows that for large cache spaces 
it is beneficial to jointly serve multiple flows if their data item sizes and popularity distributions are similar and their arrival rates do not differ significantly.
Although the optimal static resource separation scheme has been shown to always theoretically achieve the best performance under 
certain assumptions~\cite{dehghan2017sharing}, 
in practice the number of separate clusters deployed in service, e.g., Memcached clusters, is relatively small~\cite{JiangSong}.  This may be partially attributed to the self-organizing behavior of LRU for resource pooling.  As shown in Theorem~\ref{lemma:samealpha},  resource pooling can adaptively achieve the optimal resource allocation for multiple competing flows asymptotically when the data item sizes are equal. In practice two data items can be considered to have an approximately equal size if within the same range.  This property is especially beneficial when the request statistics, including the distributions and rates,  are time-varying. Nevertheless, 
we also point that, due to the impact of competing flows,  careful separation could be necessary if we want to guarantee the miss ratios of individual flows for certain QoS requirements. 

The current practice uses applications and domains to separate flows of requests into different cache spaces~\cite{JiangSong,memcachNSDI}.  One possible explanation is that the data item sizes, e.g., text and image objects, and request rates are quite different.  This also suggests that there is still room for possible improvement.
A more careful strategy based on the quantitative characterization may lead to optimal or near-optimal performance. 
For example, our analysis shows that distributions are important in determining the miss ratios. Thus, it appears to be beneficial if the statistics of different flows can be further exploited in practice.

\input{experiment2}

 \section{Conclusion}\vspace{-0.0cm}
 When designing a caching system shared by multiple request flows, should we use resource pooling or separation for better hit ratios?  This paper develops a theoretical framework to answer this fundamental question.  Roughly speaking, for flows with similar request distributions and data item sizes, with close arrival rates,  and/or with enough overlapped data items, 
 it is beneficial to jointly serve them by combining their allocated cache spaces together.  
However,  for flows with {disparate request distributions}, i.e., 
probability tails decaying at different rates,  or with clearly different arrival rates, isolating the cache spaces provides a guarantee for the hit ratios of individual flows. 
Otherwise, some of the flows could be negatively impacted, even severely penalized.  Our results provide useful insights that can be exploited to potentially further improve the hit ratios of caching systems.

\section{Proofs}\label{s:proof}
This section contains the details of the proofs.  
\subsection{Proof of Lemma~\ref{lemma:decomposition}}
Without loss of generality, we can assume that $\left(q^{(k)}_{i}\right)$ is a non-increasing sequence in $i$ for each fixed $k$. 
We begin with $g(x)=\bar{m}(x)$. 
First, using the inequality 
$(1-q^{(k)}_i)^z \leq \exp(-q^{(k)}_i z)$,
we obtain
\begin{align}\label{eq:decomp1}
 m^{(k)}(z) 
 &\geq \sum_{i=0}^{\infty} s^{(k)}_i\left( 1- \exp\left(-q^{(k)}_i z\right) \right) = \bar{m}^{(k)}(z).
\end{align}
Next, for any $0<\delta<1$,  there exists $x_{\delta}>0$ such that $1-x\geq e^{-(1+\delta)x}, 0\leq x \leq x_{\delta}$.
Thus, selecting $i_{\delta}$ with $q^{(k)}_{i_{\delta}}<x_{\delta}$, we have
\begin{align}\label{eq:decomp2}
 m^{(k)}(z) &\leq  \left( \sum_{i=0}^{i_{\delta}} + \sum_{i>i_{\delta}} \right) s^{(k)}_i\left( 1- \left(1-q^{(k)}_i\right)^z \right) \nonumber\\
 &\leq i_{\delta}\bar{s} + \sum_{i>i_{\delta}} s^{(k)}_i\left( 1- \exp\left(-(1+\delta)q^{(k)}_i z\right) \right) \nonumber\\
 &\leq i_{\delta} \bar{s}+ \bar{m}^{(k)}((1+\delta)z),
\end{align}
where the second last inequality uses $s^{(k)}_i \leq \bar{s}$.
Using (\ref{eq:decomp1}), (\ref{eq:decomp2}) and 
$\varlimsup_{x\to \infty}m^{(k)}((1+\delta) x)/m^{(k)}(x)$ $\to 1$ as $\delta\to 0$, 
we prove (\ref{eq:approx}).

Based on the representation of $\left(p^{\circ}_i, i\geq 1\right)$ in~(\ref{eq:unordered}), we have
\begin{align}\label{eq:decomp3}
m(z) &= \sum_{i=0}^{\infty} s_i\left( 1- (1-p^{\circ}_i )^z \right) \nonumber\\
 & = \sum_{k=1}^{\infty} \sum_{i=1}^{\infty}  s^{(k)}_i\left( 1- \left(1-\nu_k q^{(k)}_i\right)^z \right).
   \end{align}
 Using the same arguments as in (\ref{eq:decomp1}) and (\ref{eq:decomp2}), we can prove   
   \begin{align}\label{eq:decomp4}
  \sum_{i=1}^{\infty}  &s^{(k)}_i\left( 1- \left(1-\nu_k q^{(k)}_i\right)^z \right) \nonumber\\
  & \sim \sum_{i=0}^{\infty} s^{(k)}_i\left( 1- \exp\left(-q^{(k)}_i \nu_k z\right) \right) = \bar{m}^{(k)}(\nu_k z). 
   \end{align}
   Using (\ref{eq:decomp3}), (\ref{eq:decomp4}) and applying~(\ref{eq:approx}), we finish the proof of~(\ref{eq:decomp}) when $g(x)=\bar{m}(x)$. 
  Slightly modifying the preceding arguments can prove (\ref{eq:decomp}) when $g(x)=m(x)$.

\subsection{Proof of Theorem~\ref{theorem:missP}} \label{ss:p:theorem:missP}
In order to prove Theorem~\ref{theorem:missP}, we need to establish a lemma. 
\begin{lemma}\label{lemma:help1}
For $\epsilon(x) = \epsilon \delta(x)$ as in~(\ref{eq:mz}) and $\bar{s}=\sup_i s_i < \infty$, we obtain
\begin{align}\label{eq:Mn}
  \Pr\left[ M( m^{\leftarrow}(x)) \geq (1+\epsilon(x)) x \right] \leq e^{-(\epsilon(x))^2 x/4\bar{s}}.
\end{align}
\end{lemma}
\begin{proof}
Define a Bernoulli random variable $X_i$, and let $X_i=1$ to indicate that item $d_i$ has been requested in $R_{-1}, R_{-2}, \cdots, R_{-n}$ and $X_i=0$ otherwise. 
By Markov's inequality, for $\theta>0$, we obtain, using $\Pr[X_i=1] = p_i(n)$, 
$\expect[e^{\theta s_i X_i}]=p_i(n) e^{\theta s_i} + 1-p_i(n) = p_i(n) \left(e^{\theta s_i}-1\right) + 1 \leq e^{p_i(n) \left(e^{\theta s_i}-1\right)}$ and independence of $X_i$'s, 
{\small
\begin{align}
& \Pr \left[ M(n)\geq (1+\epsilon(m(n)))m(n) \right] \nonumber\\
&\leq \expect \left[ e^{\theta \sum_{i=1}^{\infty}s_i X_i} \right] \mathlarger{/} e^{(1+\epsilon(m(n))) \theta m(n)} \nonumber\\
& \leq  \exp \left( \sum_{i=1}^{\infty}  p_i(n) \left( e^{\theta s_i}  - 1\right) - \theta (1+\epsilon(m(n))) \sum_{i=1}^{\infty}p_i(n) s_i \right).  \nonumber
\end{align}
}
\\
\noindent Using $e^x-1\leq (1+\xi)x, 0<x<2\xi/e^{\xi}, \xi>0$, for $\theta=\epsilon(m(n))/(2\bar{s})$, we obtain,  
  $e^{\theta s_i}  - 1 \leq  (1+\epsilon(m(n)) / 2)  \theta s_i$. 
Therefore, we have
\begin{align}
\Pr  \left[ M(n)\right. & \left. \geq (1+\epsilon(m(n)) )m(n) \right] \nonumber\\
&\leq \exp \left( - \sum_{i=1}^{\infty}  \frac{(\epsilon (m(n)) ^2}{4 \bar{s}}  p_i(n) s_i \right), \nonumber
\end{align}
which, by $\sum_{i=1}^{\infty} p_i(n)s_i = m(n)$, yields
$$ \Pr\left[ M(n) \geq (1+\epsilon(m(n))) m(n) \right] \leq e^{-(\epsilon(m(n)))^2 m(n)/4\bar{s}},$$
implying (\ref{eq:Mn}) by replacing $x=m(n)$.
Using the same approach, we can prove
\begin{align}\label{eq:Mnlow}
 \Pr\left[ M( m^{\leftarrow}(x)) \leq (1-\epsilon(x)) x \right] \leq e^{-(\epsilon(x))^2 x/4\bar{s}}.
\end{align}
\end{proof}

Next we prove Theorem~\ref{theorem:missP}.
\begin{proof}
In the intuitive proof of Section~\ref{s:model}, we have derived
\begin{align}\label{eq:representation2}
  \Pr[C_0>x | I_0=k] = \Pr[\sigma >  M^{\leftarrow}(x) | I_0=k]. 
\end{align}
The whole proof consists of two steps.  
The first step is to show
\begin{align}\label{eq:firstStep}
\Pr[\sigma >  n | I_0=k] \sim \Gamma(\beta_k+1)/\Phi_k(n),
\end{align}
for $\beta_k>0$ and $\beta_k=0$,  respectively.
 The second step 
is to relate $M^{\leftarrow}(x)$ to $m^{\leftarrow}(x)$ as $x \to \infty$. 
\vspace{3mm}

\noindent \textbf{Step 1.}
First, we prove (\ref{eq:firstStep}) for $\beta_k>0$.
Assume that $\Phi_k(x)$ is eventually absolutely continuous and strictly
monotone,  since, by Proposition~1.5.8 of~\cite{regularVariation}, we can construct such 
a function
\begin{equation}\label{eq:modified}
\Phi^{\ast}_k(x)= \beta_k \int_{x_0}^{x}\Phi_k(s)s^{-1}ds, \; x\geq x_0,
\end{equation}
which, for $x_0$ large enough,  satisfies, as $y \to \infty$,
$$  \left(\sum_{i=y}^{\infty} q^{(k)}_i\right)^{-1}  \sim \Phi_k \left( \left(p^{(k)}_y\right)^{-1} \right) \thicksim \Phi_k^{\ast} \left( \left(p^{(k)}_y\right)^{-1} \right).$$
Therefore,  there exists $x_{0}$
such that for all $x>x_0$, $\Phi_k(x)$ has an inverse function
$\Phi_k^{\leftarrow}(x)$.  The condition (\ref{eq:relation}) implies
that, for $0<\epsilon_1<1$,  there exists $i_{\epsilon_1}$, such that
for $i>i_{\epsilon_1}$,
{
\begin{align}\label{phiineq} 
(1-\epsilon_1)   \left(\sum_{j=i}^{\infty} q^{(k)}_j\right)^{-1} & \leq \Phi_k \left( \left(p^{(k)}_i\right)^{-1} \right) \nonumber\\
&\leq (1+\epsilon_1)  \left(\sum_{j=i}^{\infty} q^{(k)}_j\right)^{-1}, 
\end{align}
}
and thus, by choosing $i_{\epsilon_1}$ such that $1/p^{(k)}_{i_{\epsilon_1}} > x_0$, we obtain
\begin{align}\label{eq:equivBound2}
\Phi^{\leftarrow}_k   & \left( (1-\epsilon_1) \left( \sum_{j=i}^{\infty} q^{(k)}_j \right)^{-1} \right)  \leq
\left(p^{(k)}_i\right)^{-1} \nonumber \\
&\leq  \Phi_k^{\leftarrow} \left((1+\epsilon_1) \left( \sum_{j=i}^{\infty} q^{(k)}_j \right)^{-1} \right).
\end{align}

First, we will prove an \emph{upper bound} for (\ref{eq:firstStep}).
Combining (\ref{eq:rep_in}) and  (\ref{eq:equivBound2}) yields,  using $(1-p)^n\leq e^{-np}$, 
 \begin{align}\label{eq:repUpper}
  \Pr &\left[\sigma>n \mathlarger{\mid} I_0=k\right] 
   = \left( \sum_{i=1}^{i_{\epsilon_1}} 
  + \sum_{i=i_{\epsilon_1}+1}^{\infty} \right) q_i^{(k)}\left(1-p_i^{(k)} \right)^n  \nonumber\\
  & \leq \left(1-p_{i_{\epsilon_1}}^{(k)} \right)^n +  \sum_{i=i_{\epsilon_1}+1}^{\infty} q_i^{(k)}e^{ -n p_i^{(k)}} \eqdef I_1+I_2.
\end{align}
For $0< \epsilon_2 \leq p_{i_{\epsilon_1}}^{(k)}$, integer $n$ large enough, and any nonnegative integer $l \leq \lfloor \log n\epsilon_2 \rfloor$, 
we can find $i_l$ such that $p_{i_l+1}^{(k)} \leq e^l/n \leq p_{i_l}^{(k)} \leq \epsilon_2$. 
Choose an integer $m$ with $0<m <\lfloor \log n\epsilon_2 \rfloor$.  We have $i_0 > i_m > i_{\lfloor \log n\epsilon_2 \rfloor} > i_{\epsilon_1}$, and
\begin{align}\label{eq:I2}
I_2 & = \left( \sum_{i=i_{\epsilon_1}+1}^{i_{\lfloor \log n\epsilon_2 \rfloor}-1}  + \sum_{i={i_{\lfloor \log n\epsilon_2 \rfloor}}}^{i_m} 
 +  \sum_{i=i_m+1}^\infty  \right) q_i^{(k)} e^{-np_i^{(k)}} \nonumber \\
 & \leq e^{-n\epsilon_2} + \sum_{l=m}^{\infty} e^{-e^l} \sum_{j=i_{l+1}+1}^{i_l} q_j^{(k)} + \sum_{j=i_m+1}^\infty q_{j}^{(k)}e^{-np_j^{(k)}} \nonumber \\
  &\eqdef I_{21}+I_{22}+I_{23}.
\end{align}
We have $I_{23}  =  \sum_{j=i_m+1}^\infty \left( Q_{j} - Q_{j+1}\right)$ $e^{-n/\Phi_k^\leftarrow \left((1+\epsilon_1)Q_j^{-1}\right)}$ for $Q_i = \sum_{j=i}^\infty q_j^{(k)}$.
Since $e^{-n/\Phi_k^\leftarrow \left((1+\epsilon_1)u^{-1}\right)}$ is decreasing with $u$, 
we obtain $e^{-n/\Phi_k^\leftarrow \left((1+\epsilon_1)u^{-1}\right)} \geq e^{-n/\Phi_k^\leftarrow \left((1+\epsilon_1)Q_j^{-1}\right)}$ for $\forall u \in (Q_{j+1} , Q_{j})$,
which implies
\begin{align}
 I_{23} 
  & \leq \int_0^{Q_{i_m}} e^{-n/\Phi_k^\leftarrow \left((1+\epsilon_1)u^{-1}\right)} du \nonumber \\
  & \leq    \int_{0}^{\epsilon_1} d\left( \frac{1+\epsilon_1}{\Phi_k \left( n/z \right)  } \right)  +  \int_{\epsilon_1}^{e^m} e^{-z} d\left( \frac{1+\epsilon_1}{\Phi_k \left( n/z \right)  } \right). \nonumber 
  \end{align}
By Theorem 1.2.1 of \cite{regularVariation} and (\ref{eq:modified}), 
we obtain,
\begin{equation}\label{eq:propupper5I1}
  I_{23} \Phi_k(n) \lesssim    (1+\epsilon_1) \epsilon_1 ^ {\beta_k} + \int_{\epsilon_1}^{e^m}(1+\epsilon_1)\beta_k e^{-z} z^{\beta_k-1}dz  .
\end{equation}
For $I_{22}$, using the same approach, we obtain
\begin{equation}\label{eq:propupper6} 
I_{22} \Phi_k(n) \lesssim  \sum_{k=m}^{\infty} (1+\epsilon_1)e^{-e^{k}} \left(e^{k+1} \right)^{\beta_k}<\infty.
\end{equation}
Combining (\ref{eq:propupper5I1}) and (\ref{eq:propupper6}), and then
 passing $\epsilon_1 \to 0$ and $m \to \infty$,  we obtain, using $I_1=o(1/\Phi_k(n))$ in (\ref{eq:repUpper}),
  \begin{align}\label{eq:asympN5}
   \Pr[\sigma>n | I_0=k]\Phi_k(n) &\lesssim \int_{0}^{\infty}\beta_k e^{-z}z^{\beta_k-1}dz \nonumber \\
&=\Gamma(\beta_k+1).
\end{align}

Now, we prove the \emph{lower bound} for  (\ref{eq:firstStep}).
By condition (\ref{eq:1class}) we can choose $i_{\epsilon_1}$ large enough such that, for all $i>i_{\epsilon_1}$,
\begin{align}\label{eq:class1upper}
 q_i^{(k)} \geq (1-\epsilon_1)q_{i-1}^{(k)} .
\end{align}
Using~(\ref{eq:rep_in}), (\ref{eq:equivBound2}) and the monotonicity of $\Phi_k^{\leftarrow}(\cdot)$, we obtain 
{\small
 \begin{align}\label{eq:low1}
  &\Pr\left[\sigma>n \mathlarger{\mid} I_0=k\right]  \geq \sum_{i=i_{\epsilon_1}+1}^{\infty} q_i^{(k)}\left(1-p_i^{(k)} \right)^n \nonumber\\
  & \geq   (1-\epsilon_1)  \sum_{i=i_{\epsilon_1}+1}^{\infty} \Delta Q_i \left(1-\frac{1}{ \Phi^{\leftarrow}_k  \left( (1-\epsilon_1) \left(Q_i \right)^{-1} \right)} \right)^n \nonumber \\
  & \geq   (1-\epsilon_1)  \int_{0}^{Q_{i_{\epsilon_1}}} \left(1-1 \mathlarger{\mathlarger{/}} \Phi^{\leftarrow}_k  \left( (1-\epsilon_1) u^{-1} \right) \right)^ndu.
  \end{align}
}
\\
 \noindent where $\Delta Q_i = Q_{i-1}-Q_i=q_{i-1}^{(k)}$. For $W>0$, choosing $i_{n}>i_{\epsilon}$ with $ \Phi^{\leftarrow}_k\left((1-\epsilon_1)Q_{i_n}\right) = n/W$ 
 and letting $z=n/\Phi^{\leftarrow}_k\left( (1-\epsilon_1) u^{-1} \right)$,  we obtain, 
\begin{align}\label{eq:asympN6}
   \Pr&[ \sigma>n | I_0=k]\Phi_k(n) \nonumber\\
   & \geq  (1-\epsilon_1)\Phi_k(n)\int_{0}^{Q_{i_n}} \left( 1- \frac{1}{\Phi^{\leftarrow}_k\left( (1-\epsilon_1)/u \right)} \right)^ndu \nonumber\\
        & \geq  (1-\epsilon_1) \int_{\epsilon_1}^{W}  \left(1- \frac{z}{n}\right)^n d \left( \frac{1-\epsilon_1}{\Phi_k(n/z)} \right). 
\end{align}
From (\ref{eq:asympN6}), by using the same approach as in deriving
(\ref{eq:propupper5I1}), we obtain, as $n \to \infty$,
\begin{equation}
\Pr[ \sigma>n | I_0=k] \Phi_k(n) \gtrsim  (1-\epsilon_1)\int_{\epsilon}^{W}(1-\epsilon_1)\beta_k e^{-z}
z^{\beta_k-1}dz, \nonumber
\end{equation}
which, passing  $W \to \infty$ and $\epsilon_1 \to 0$,  yields
  \begin{align}\label{eq:asympN7}
   \Pr[ \sigma>n | I_0=k] \Phi_k(n) & \gtrsim   \int_{0}^{\infty}\beta_k e^{-z} z^{\beta_k-1}dz \nonumber \\
& = \Gamma(\beta_k+1).
\end{align}
Combining (\ref{eq:asympN5}) and (\ref{eq:asympN7}) completes the
proof of (\ref{eq:firstStep}) for $\beta_k>0$.

For $\beta_k=0$, we need to prove
\begin{align}\label{sgmeq1}
 \Pr \left[ \sigma > n \mathlarger{\mid} I_0=k\right] \sim 1/\Phi_k(n). 
\end{align}
Recall that there exists $x_0$ such that $\Phi_k(x)$ is strictly increasing for $x>x_0$.
For any positive integer $n$, we can find $\epsilon_3 \in (0,1)$ with $n/\epsilon_3>x_0$. 
Because of the monotonicity of $p_i^{(k)}$, there exists an index $i_{\epsilon_3}$  such that $p_{i_{\epsilon_3}}^{(k)} \geq \epsilon_3/n$ and $p_i^{(k)} < \epsilon_3/n$ for all $i>i_{\epsilon_3}$.
By choosing $\epsilon_3$ sufficiently small such that $i_{\epsilon_3}>i_{\epsilon_1}$, we can derive the lower bound for the probability
\begin{align}
\Pr \left[ \sigma > n \mathlarger{\mid} I_0=k\right] & = \sum_{i=1}^{\infty} q_i^{(k)}\left(1-p_i^{(k)} \right)^n \nonumber\\
& \geq \left(1-\frac{\epsilon_3}{n} \right)^n \sum_{i=i_{\epsilon_3}+1}^{\infty} q_i^{(k)} . \nonumber 
\end{align}
Using (\ref{phiineq}) and (\ref{eq:class1upper}), we obtain
\begin{align}
 \Pr \left[ \sigma > n \mathlarger{\mid} I_0=k\right] & \geq \left(1-\frac{\epsilon_3}{n} \right)^n (1-\epsilon_1)\sum_{i=i_{\epsilon_3}}^{\infty} q_i^{(k)}  \nonumber\\
 &   \geq \left(1-\frac{\epsilon_3}{n} \right)^n \frac{{(1-\epsilon_1)}^2}{\Phi_k\left( \left( p_{i_{\epsilon_3}}^{(k)}\right)^{-1} \right)} \nonumber \\
 & \geq \left(1-\frac{\epsilon_3}{n} \right)^n \frac{{(1-\epsilon_1)}^2}{\Phi_k\left( \frac{n}{\epsilon_3} \right)},   \nonumber 
\end{align} 
implying 
\begin{align*}
\Pr \left[ \sigma > n \mathlarger{\mid} I_0=k\right] \Phi_k(n) \geq \left(1-\frac{\epsilon_3}{n} \right)^n {(1-\epsilon_1)}^2 \frac{\Phi_k(n)}{\Phi_k\left( \frac{n}{\epsilon_3} \right)}.
\end{align*}
By passing $n \to \infty$ and $\epsilon_1, \epsilon_3 \to 0$, we obtain
\begin{align}
 \varliminf_{n \to \infty}  \Pr \left[ \sigma > n \mathlarger{\mid} I_0=k\right] \Phi_k(n) \geq 1 . \label{lwbd1}
\end{align}

Next, we prove the upper bound. The probability $\Pr [ \sigma > n | I_0=k]$ can be bounded as
\begin{align}
\Pr & \left[ \sigma > n \mathlarger{\mid} I_0=k\right]  = \left( \sum_{i=i_{\epsilon_1}+1}^{\infty} + \sum_{i=1}^{i_{\epsilon_1}} \right) q_i^{(k)}\left(1-p_i^{(k)} \right)^n  \nonumber \\
& \hspace{0.3cm} \leq  \sum_{i=i_{\epsilon_1}+1}^{\infty} q_i^{(k)} e^{-np_i^{(k)}}+ \left(1-p_{i_{\epsilon_1}}^{(k)} \right)^n, \label{upbdineq}
\end{align}
where the second inequality uses the monotonicity of $p_i^{(k)}$ and $1-x \leq e^{-x}$ .
Using a similar approach as in (\ref{eq:I2}), we can upper bound (\ref{upbdineq}) by
{\small
\begin{align}
&\left ( \sum_{i=i_{\epsilon_1}+1}^{i_{\lfloor \log n\epsilon_2 \rfloor}-1} 
  + \sum_{i={i_{\lfloor \log n\epsilon_2 \rfloor}}}^{i_m}  +  \sum_{i=i_m+1}^\infty \right) q_i^{(k)} e^{-np_i^{(k)}}  + \left(1-p_{i_{\epsilon_1}}^{(k)} \right)^n \nonumber \\
 & \leq e^{-n\epsilon_2} + \sum_{l=m}^{\infty} e^{-e^l} \sum_{j=i_{l+1}+1}^{i_l} q_j^{(k)} + \sum_{i=i_m+1}^\infty q_i^{(k)} + \left(1-p_{i_{\epsilon_1}}^{(k)} \right)^n \nonumber \\
 & \leq e^{-n\epsilon_2} + \sum _{l=m}^{\infty} e^{-e^l} \frac{(1+\epsilon_1)}{\Phi_k (n/e^{l+1})} + \frac{1+\epsilon_1}{\Phi_k (n/e^m)} + \left(1-p_{i_{\epsilon_1}}^{(k)} \right)^n, \nonumber 
\end{align}
}
\hspace{-0.2cm} which implies
\begin{align}
 \Pr &\left[ \sigma > n \mathlarger{\mid} I_0=k\right] \Phi_k(n) \leq  \frac{(1+\epsilon_1)\Phi_k(n)}{\Phi_k (n/e^m)} \nonumber \\
 & + \sum _{l=m}^{\infty} e^{-e^l} \frac{(1+\epsilon_1)\Phi_k(n)}{\Phi_k (n/e^{l+1})} + o(1) . \label{pphiupbd1}
\end{align}
Passing $\epsilon_1\to 0$, $n \to \infty$ and then $m \to \infty$ in (\ref{pphiupbd1}) yields
\begin{align}
 \varlimsup_{n \to \infty} \Pr \left[ \sigma > n \mathlarger{\mid} I_0=k\right] \Phi_k(n) \leq 1. \label{upbd1} 
\end{align}
Combining (\ref{lwbd1}) and (\ref{upbd1}) finishes the proof of (\ref{sgmeq1}).

Up to now, we have proved (\ref{eq:firstStep}) for $\beta_k\geq 0$.
Next, we use the concentration bounds~(\ref{eq:Mn}) and (\ref{eq:Mnlow}) for $M(x)$ in Lemma~\ref{lemma:help1} to characterize $M^{\leftarrow}(x)$. 

\noindent \textbf{Step 2.}
For  $x_1= m^{\leftarrow}(x/(1+\epsilon(x)))$, we obtain, by~(\ref{eq:Mn}),
\begin{align}\label{eq:Mx1a}
  \Pr&[M^{\leftarrow}(x) <  x_1]  \leq \Pr\left[ M( m^{\leftarrow}(x/(1+\epsilon(x)))) \geq x\right] \nonumber\\
    &\vspace{0.0cm}= \Pr\left[ M\left( m^{\leftarrow}\left(\frac{x}{1+\epsilon(x)} \right)\right) \geq  \left( \frac{(1+\epsilon(x))x}{1+\epsilon(x)}\right) \right].
  \end{align}
Recalling $h_1$ and $h_2$ defined in~(\ref{eq:delta}) and noting $\delta(x)\leq 1$,  we have, for $x>x_0$, 
$\epsilon(x) \geq h_1 \epsilon\left(x/(1+\epsilon(x)) \right)$,
which, in conjunction with~(\ref{eq:Mx1a}) and using~(\ref{eq:Mn}), implies that $\Pr[M^{\leftarrow}(x) <  x_1]$ is upper bounded by
\begin{align}\label{eq:Mx1b}
  \Pr&\left[ M\left( x_1 \right) \geq \left(1+ h_1 \epsilon\left(\frac{x}{1+\epsilon(x)} \right) \right) \left( \frac{x}{1+\epsilon(x)}\right) \right] \nonumber\\
  &\hspace{0.3cm}\leq \exp\left(-(h_1\epsilon(x/(1+\epsilon(x))))^2 x/\left(4\bar{s} (1+\epsilon(x))\right)\right) \nonumber\\
  &\hspace{0.3cm}\leq \exp\left( -\frac{h_1^2}{h_2^2} \frac{\epsilon(x)^2 x}{ 4\bar{s} (1+\epsilon)} \right).
\end{align}
Thus, by (\ref{eq:representation2}), (\ref{eq:firstStep}), (\ref{eq:Mn}) and (\ref{eq:Mx1b}), we obtain
\begin{align}
  & \Pr [C_0>x| I_0=k] \leq \Pr[\sigma >  M^{\leftarrow}(x), M^{\leftarrow}(x)\geq  x_1 | I_0=k]  \nonumber\\
  & \hspace{29mm} + \Pr[M^{\leftarrow}(x) <  x_1] \nonumber\\
  & \leq \Pr[ \sigma >  m^{\leftarrow}(x/(1+\epsilon(x))) | I_0=k] + \Pr[M^{\leftarrow}(x) <  x_1] \nonumber \\
  &\lesssim \frac{\Gamma(\beta_k+1)}{\Phi_k\left( m^{\leftarrow}(x/(1+\epsilon(x)))\right)} 
    +\exp\left( -\frac{h_1^2}{h_2^2} \frac{\epsilon(x)^2 x}{ 4\bar{s} (1+\epsilon)} \right). \nonumber
\end{align}

Using $\varlimsup_{x\to \infty} \log \left(m^{\leftarrow}(x)\right)/(\delta^2(x) x) =0$ and (\ref{eq:mz}), we obtain, recalling $\epsilon(x)=\epsilon \delta(x)$ and passing $\epsilon\to 0$,
\begin{align}\label{eq:representation3}
 \Pr [C_0 &>x | I_0=k] \nonumber \\
 & \leq \frac{\Gamma\left(1+\beta_k\right)}{\Phi_k\left(m^{\leftarrow}(x) \right)}  + o\left( 1/\Phi_k\left( m^{\leftarrow}(x) \right) \right).
\end{align}
Let $x_2=m^{\leftarrow}(x/(1-\epsilon(x)))$. We obtain
\begin{align}
  \Pr [C_0>x|I_0=k] & \geq \Pr[\sigma > M^{\leftarrow}(x), M^{\leftarrow}(x)\leq  x_2 | I_0=k]  \nonumber \\
  & - \Pr[M^{\leftarrow}(x) >  x_2], \nonumber
\end{align}
which, by similar arguments as in proving (\ref{eq:representation3}),  yields
\begin{align}\label{eq:representation4}
 \Pr [C_0 & >x|I_0=k]  \nonumber \\
 & \geq \frac{\Gamma\left(1+\beta_k\right)}{\Phi_k\left( m^{\leftarrow}(x) \right)}  - o\left( 1/\Phi_k\left( m^{\leftarrow}(x) \right) \right).
\end{align}
Combining (\ref{eq:representation3}) and (\ref{eq:representation4}) finishes the proof.
\end{proof}

\subsection{Proof of Corollary~\ref{corollary:zipf2}} \label{ss:p:corollary:zipf2}
Consider $p_x^{(k)} \sim l(x)/x^{\alpha}$ with $l(x)$ being a slowly varying function.
According to Proposition~$1.5.10$ of \cite{regularVariation}, we have 
\begin{align}
 \Pr\left[R_0>x\right] = \sum_{i\geq x}p_i^{(k)} \sim \int_{x}^{\infty} \frac{l(x)}{x^{\alpha}}dx \sim \frac{l(x)}{(\alpha-1)x^{\alpha-1}}. \label{eq:missP1P2}
\end{align}
Using Lemma~\ref{lemma:decomposition}, we obtain
\begin{align}
  m(z) & \sim \sum_{i\geq 1} \left( 1-\exp\left(-\frac{l(i)z}{i^\alpha}\right) \right) \nonumber\\
  & \sim  \int_1^\infty \left( 1-\exp\left(-\frac{l(x)z}{x^\alpha}\right) \right)dx. \nonumber
\end{align}
Since $l(x)x^{-\alpha} \sim \alpha \int_x^\infty l(t)t^{-\alpha-1}dt$ (Proposition~1.5.10 of \cite{regularVariation}), 
for any $\epsilon>0$, there exists $x_\epsilon>0$, such that for all $x>x_\epsilon$, 
\begin{align}\label{ineq:lzipf}
 (1-\epsilon) \alpha &\int_x^\infty l(t)t^{-\alpha-1} < l(x)x^{-\alpha} \nonumber\\
 &< (1+\epsilon) \alpha \int_x^\infty l(t)t^{-\alpha-1}dt. 
\end{align}
Therefore, $m(z)$ can be upper bounded by 
\begin{align}
 &m(z) \lesssim  \int_1^{x_\epsilon} \left( 1-\exp\left(-\frac{l(x)z}{x^\alpha}\right) \right)dx \nonumber\\
 & + \int_{x_\epsilon}^\infty \left( 1-\exp\left(-(1+\epsilon) \alpha z \int_x^\infty l(t)t^{-\alpha-1}dt\right)\right)dx. \nonumber
\end{align}

Define $f(x)=\alpha \int_x^\infty l(t)t^{-\alpha-1}dt$. 
We obtain
\begin{align}
 & m(z)  \lesssim x_\epsilon+ \int_{x_\epsilon}^\infty \left( 1-e^{- (1+\epsilon) z f(x)} \right)dx \nonumber \\
  & = x_\epsilon + x \left. \left(1-e^{- (1+\epsilon) z f(x)} \right) \right|_{x_\epsilon}^\infty + \int_{x_\epsilon}^\infty x de^{-(1+\epsilon) z f(x)} \nonumber \\
  & = x_\epsilon^{-(1+\epsilon)zf(x_\epsilon)} + \int_{x_\epsilon}^\infty x de^{-(1+\epsilon) z f(x)} 
   \eqdef I_1 + I_2. \label{eq:mI1I2}
\end{align}
For $y=f(x)$, we have
\begin{align} \label{eq:I2P2}
 I_2 & = \int_{x_\epsilon}^\infty x e^{-(1+\epsilon) z f(x)} (-(1+\epsilon) z f^\prime (x)) dx \nonumber\\
  &= \int_0^{f(x_\epsilon)}  (1+\epsilon) z f^\leftarrow(y) e^{- (1+\epsilon) z y}dy,
\end{align}
where $f^\leftarrow$ is the inverse function of $f$.
Let $g(x)=1/x$, $h(x)=g\comp f(x)$, and $l_1(x)={l(x)}^{-1/\alpha}$. We have
$h(x) \sim x^\alpha {l_1}^\alpha(x)$. 
By Proposition~1.5.15 of~\cite{regularVariation}, we obtain the asymptotic inverse of $h$, 
\begin{align}
 h^\leftarrow (x) \sim x^{1/\alpha}l_1^\# (x^{1/\alpha}). \nonumber
\end{align}
Recall $l_{n+1}(x) \eqdef l_1(x/l_n(x)), n=1, 2,\ldots$ and $l_n(x) \sim l_{n+1}(x)$ as $x \to \infty$ for some $n\geq2$. 
Using Proposition~2.3.5 of~\cite{regularVariation}, we have $l_1^\# \sim 1/l_n(x)$.
Therefore, 
\begin{align}
 h^\leftarrow (x) \sim x^{1/\alpha}/l_n (x^{1/\alpha}). \nonumber
\end{align}
Since $h=g \comp f$, we have 
\begin{align}
 h^\leftarrow(x) = f^\leftarrow (g^\leftarrow (x)) & = f^\leftarrow (1/x). \nonumber
\end{align}
implying, as $x \to 0$,
\begin{align}
 f^\leftarrow (x) = h^\leftarrow (1/x) \sim \frac{1}{x^{1/\alpha}l_n (x^{-1/\alpha})}. \nonumber 
\end{align}
For $y=f(x_\epsilon)$ small enough, we have
\begin{align}
  \frac{1-\epsilon}{y^{1/\alpha}l_n (y^{-1/\alpha})} < f^\leftarrow (y) < \frac{1+\epsilon}{y^{1/\alpha}l_n (y^{-1/\alpha})}. \label{eq:finverseup}
\end{align}
Combining (\ref{eq:I2P2}) and (\ref{eq:finverseup}) yields 
\begin{align}
 I_2 & < {(1+\epsilon)}^2 z \int_0^{f(x_\epsilon)} \frac{e^{- zy}}{y^{1/\alpha}l_n (y^{-1/\alpha})} dy \nonumber\\
      & <  {(1+\epsilon)}^2 z \int_0^\infty \frac{e^{- zy}}{y^{1/\alpha}l_n (y^{-1/\alpha})} dy \nonumber.
\end{align}
Using Theorem~$1.7.1^\prime$ of \cite{regularVariation}, we obtain, as $z \to \infty$,
\begin{align}
  z \int_0^\infty \frac{e^{- zy}}{y^{1/\alpha}l_n (y^{-1/\alpha})} dy
  \sim \Gamma(1-1/\alpha) z^{1/\alpha} / l_n(z^{1/\alpha}), \label{eq:LaplaceT}
\end{align}
implying
\begin{align}
 I_2 < {(1+\epsilon)}^3  \Gamma(1-1/\alpha) z^{1/\alpha} / l_n(z^{1/\alpha}). \nonumber
\end{align}
Therefore, using (\ref{eq:mI1I2}),
we have, for $z$ large enough,
\begin{align}
 m(z) \lesssim & {(1+\epsilon)}^3  \Gamma(1-1/\alpha) z^{1/\alpha} / l_n(z^{1/\alpha}) \nonumber\\
 &+ x_\epsilon^{-(1+\epsilon)zf(x_\epsilon)}. \label{ineq:m_upperbound}
\end{align}

Next, we prove a lower bound for $m(z)$.
Recalling (\ref{ineq:lzipf}) and using a similar approach as in~(\ref{eq:mI1I2}), we have
\begin{align}
 m(z) 
  & \gtrsim \int_{x_\epsilon}^\infty \left( 1-e^{- (1-\epsilon) z f(x)} \right)dx \nonumber\\
   &=  x \left. \left(1-e^{- (1-\epsilon) z f(x)} \right) \right|_{x_\epsilon}^\infty 
   + \int_{x_\epsilon}^\infty x de^{-(1-\epsilon) z f(x)} \nonumber \\
   & > \int_0^{f(x_\epsilon)}  (1-\epsilon) z f^\leftarrow(y) e^{- (1-\epsilon) z y}dy \nonumber \\
  & = \left( \int_0^\infty  - \int_{f(x_\epsilon)}^\infty \right) (1-\epsilon) z f^\leftarrow(y) e^{- (1-\epsilon) z y}dy \nonumber \\
  & > \int_0^\infty  (1-\epsilon) z f^\leftarrow(y) e^{- (1-\epsilon) z y}dy - x_\epsilon e^{-(1-\epsilon) z f(x_\epsilon)}. \nonumber 
\end{align}
Using (\ref{eq:finverseup}) and (\ref{eq:LaplaceT}), we obtain
\begin{align}
 m(z) &  \gtrsim {(1-\epsilon)}^2 z \int_0^{f(x_\epsilon)} \frac{e^{- zy}}{y^{1/\alpha}l_n (y^{-1/\alpha})} dy  \nonumber \\
 & \hspace{4mm} - x_\epsilon e^{-(1-\epsilon) z f(x_\epsilon)} \nonumber \\
 & \gtrsim {(1-\epsilon)}^3  \Gamma(1-1/\alpha) z^{1/\alpha} / l_n(z^{1/\alpha}) \nonumber\\
 & \hspace{4mm}- x_\epsilon e^{-(1-\epsilon) z f(x_\epsilon)}. \label{ineq:m_lowerbound}
\end{align}
Combining (\ref{ineq:m_upperbound}) and (\ref{ineq:m_lowerbound}), and passing $z \to \infty$ and then $\epsilon \to 0$, we obtain
\begin{align}
 m(z) \sim \Gamma(1-1/\alpha) z^{1/\alpha} / l_n(z^{1/\alpha}). \label{eq:m_asy}
\end{align}

Define 
\begin{align}
F(z) =  \frac{(\alpha-1)z^{\alpha-1}} {{\Gamma(1-1/\alpha)}^{\alpha-1}l(z)}. \nonumber
\end{align}
Now, we show $F(z) \sim \Phi_k(m^\leftarrow(z))$ as $z \to \infty$,
which is equivalent to $F(m(x^\alpha/l(x))) \sim \Phi(x^\alpha/l(x))$ as $x \to \infty$.
Using (\ref{eq:m_asy}), we obtain
\begin{align}
 & F(m(x^\alpha/l(x))) = \frac{\alpha-1}{{\Gamma(1-1/\alpha)}^{\alpha-1} }   \frac {{m(x^\alpha/l(x))}^{\alpha-1}} {l(m(x^\alpha/l(x))) } \nonumber \\
  & \sim  \frac{(\alpha-1) x^{\alpha-1} {l_1(x)}^{\alpha-1}} { {l_n(xl_1(x))}^{\alpha-1}}  
  \left/  l \left( \frac{\Gamma(1-1/\alpha)xl_1(x)}{l_n(xl_1(x))} \right) \right. \nonumber \\
  & \sim \frac {(\alpha-1) x^{\alpha-1} {c(x)}^{\alpha-1} }  {l(xc(x))}, \label{eq:Fm}
\end{align}
where $c(x)=l_1(x)/l_n(xl_1(x))$.

Recall $l_n(y) \sim l_{n+1}(y)=l_1(y/l_n(y)), y \to \infty$.
For $y=xl_1(x)$,  by Proposition~1.5.15 of~\cite{regularVariation}, we obtain $$x \sim y l_1^\#(y) \sim y/ l_n(y),$$
which implies
\begin{align}
 l_n(x l_1(x)) & = l_n(y) 
  \sim  l_1(y/l_n(y)) 
  \sim l_1(x).  \nonumber
\end{align}
Therefore, we obtain 
\begin{align}
 \lim_{x \to \infty} c(x) = \lim_{x \to \infty} \frac{l_1(x)}{l_n(x l_1(x))} =1.  \label{eq:c_asy}
\end{align}

Combining (\ref{eq:relation}), (\ref{eq:Fm}) and (\ref{eq:c_asy}) yields 
$$\Phi_k(x^\alpha/l(x)) \sim F(m(x^\alpha/l(x))),$$
which implies
\begin{align}
 \Phi_k(m^\leftarrow(z)) & \sim F(z)
  = \frac{(\alpha-1)z^{\alpha-1}} {{\Gamma(1-1/\alpha)}^{\alpha-1}l(z)}, \; \text{as}\, z\to \infty. \nonumber 
\end{align}

Therefore, by Theorem~\ref{theorem:missP}, we obtain, as $x\to \infty$, 
\begin{align}\label{eq:missP2P2}
 \Pr[C_0>x] \sim \frac{\Gamma(2-1/\alpha) \Gamma(1-1/\alpha)^{\alpha-1}}{\alpha-1}\frac{l(x)}{x^{\alpha-1}}.
\end{align}   
Combining~(\ref{eq:missP1P2}) and (\ref{eq:missP2P2}) finishes the proof.

\subsection{Proof of Theorem~\ref{lemma:che}} \label{ss:p:lemma:che}
The Che approximation gives,  for a LRU cache of size $x$,
 \begin{align}\label{eq:che}
 \Pr_{che}[C_0>x|I_0=k] = \sum_{i=1}^\infty q_i^{(k)} e^{-p_i^{(k)}T}, 
\end{align}
where $T$ is the characteristic time that is the unique solution to
 $\sum_{i=1}^\infty ( 1- e^{-p_i^\circ T})=x$. 
By Lemma~\ref{lemma:decomposition},
we obtain
\begin{align} \label{eq:approxT}
 T \sim m^\leftarrow(x).
\end{align}
Combining (\ref{eq:rep_in}), (\ref{eq:firstStep}) and using $e^{-y} \leq 1-y$, we derive a lower bound of (\ref{eq:che}),
\begin{align} \label{eq:chelower}
 \Pr_{che}  [C_0>x|I_0=k] & \geq \sum_{i=1}^{\infty} q_i^{(k)}\left(1-p_i^{(k)} \right)^{T} \nonumber \\
 &  = \Pr[\sigma>T | I_0=k]  \nonumber \\
 & \sim \frac{\Gamma(\beta_k+1)}{\Phi_k(T)}.
\end{align}
Next, we derive an upper bound.
Using a similar approach that proves an upper bound for $\Pr[\sigma>n|I_0=k]$ in the proof of Theorem~\ref{theorem:missP}, we have
 \begin{align}\label{eq:cheupper}
  \Pr_{che}[  & C_0  >x|I_0=k]  = \left( \sum_{i=1}^{i_{\epsilon_1}} + \sum_{i=i_{\epsilon_1}+1}^{\infty} \right)  q_i^{(k)} e^{-p_i^{(k)}T}  \nonumber \\
  & \leq e^{-p_{i_{\epsilon_1}}^{(k)}T} +  \sum_{i=i_{\epsilon_1}+1}^{\infty} q_i^{(k)} e^{-p_i^{(k)}T}  \lesssim \frac{\Gamma(\beta_k+1)}{\Phi_k(T)},
 \end{align}
 where $\epsilon_1$ is defined in (\ref{eq:equivBound2}).
Combining (\ref{eq:chelower}) and (\ref{eq:cheupper}) yields, as $x \to \infty$,
 \begin{align} \label{eq:che2}
  \Pr_{che}[C_0>x|I_0=k] \sim \frac{\Gamma(\beta_k+1)}{\Phi_k(T)}.
 \end{align}
 Combining (\ref{eq:approxT}), (\ref{eq:che2}) and using the fact 
 $$\lim_{x\to \infty} x^{\beta_k} l_k(x)/\Phi_k(x)=1,$$ 
 we complete the proof.

\bibliographystyle{abbrv}
\bibliography{mybib}

\end{document}

%% file: experiment2.tex
\section{Experiments}\label{s:experiment}

We implement an LRU simulator using C++ and conduct extensive simulation experiments.
First, we verify Theorem~\ref{theorem:missP} for data items of varying sizes with distributions beyond Zipf's distributions. 
Next, we study the interactions among multiple competing flows on the same cache space. 
Last,  we investigate the impact of overlapped data items across flows by verifying Corollaries~\ref{corollary:4} and \ref{corollary:5}.   
The numerical results based on analyses match accurately with the simulation experiments, even for small cache sizes and finitely many data items.

\begin{experiment} \label{exp:1}
Consider $3$ flows sharing a server with $10$ different item sizes. 
Set $[\nu_1,\nu_2, \nu_3]= [0.2, 0.3, 0.5]$.
Based on the empirical data size distribution~\cite{JiangSong}, we assume that the 
data sizes $\{s_i\}$ are i.i.d. random variables draw 
from a multinomial distribution with $s_i \in \{1,2,\ldots,10\}$ and parameters $[0.2, 0.15, 0.1, 0.1, 0.08, 0.09, 0.06, 0.06, 0.04, 0.02]$.
Let $d_{i,j}^{(k)}$ denote the $i'$th data item with size $j$ of flow~$k$, $1\leq k \leq 3$, $1\leq j \leq10$, $1 \leq i \leq N=10^6$.
Let $q_{i,j}^{(k)}\eqdef \Pr[R_0=d_{i,j}^{(1)} | I_0 = 1]$. Set $q_{1,j}^{(1)}=0.1, q_{1,j}^{(2)}=0.15, q_{1,j}^{(3)}=0.2$ for $1 \leq j\leq 10$,
 and $q_{i,j}^{(1)}= c_1 \log i/i^{\alpha_1}, q_{i,j}^{(2)}= c_2 \log i /i^{\alpha_2}, q_{i,j}^{(3)} = c_3 \log i/i^{\alpha_3}$ for $2 \leq i \leq N$ and $1 \leq j \leq10$. 
Set $\alpha_1=2.2$, $\alpha_2=2.4 $, $\alpha_3=2.6$. 
Then, $c_1=(1-q_{1,j}^{(1)})/(\sum_{i=2}^N (\log i) i^{-\alpha_1})=1.4193$, $c_2=(1-q_{1,j}^{(2)})/(\sum_{i=2}^N (\log i) i^{-\alpha_2})=1.8804$, $c_3=(1-q_{1,j}^{(3)})/(\sum_{i=2}^N (\log i)i^{-\alpha_3})=2.3910$.
\begin{figure}[h]
\vspace{-0.3cm}
\centering
\includegraphics[width=7.2cm]{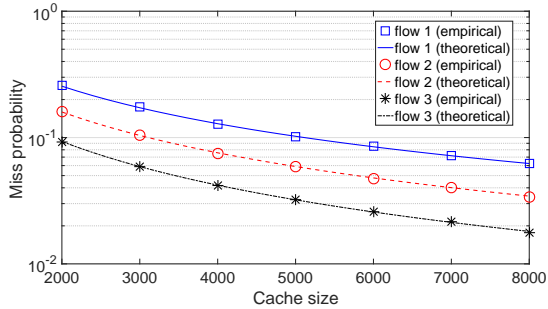}\vspace{-0.3cm}
\caption{Flows beyond Zipf with varying sizes}
\label{fig:exp1}
\vspace{-0.2cm}
\end{figure}
Theoretical miss probabilities are approximated by Theorem~\ref{theorem:missP}.  In order to improve the accuracy for a small cache size $x$, 
$m^\leftarrow(x)$ is evaluated by a numerical method based on $m(x)$. 
The empirical miss probabilities and their theoretical approximations are plotted in Fig.~\ref{fig:exp1}, which
 match very well and validate Theorem~\ref{theorem:missP} when data items have different sizes and their distributions are beyond a Zipf's distribution ($\approx \log i / i^{\alpha_k}$).
\end{experiment}

\begin{experiment}\label{exp:2}
This experiment compares the miss ratios when a flow is served exclusively in a dedicated cache 
and when it shares the same cache with other flows.  
We show how one flow is impacted by other competing flows, through validating Corollary~\ref{corollary:3}. 
Consider $10$ flows without overlapped data items. Let $\nu_k=0.1$ and
$d_i^{(k)}, i=1,2, 3, \cdots$ be the data items of flow~$k$ for $1 \leq k \leq 10$.
Data popularities of each flow are assumed to follow a Zipf's law, i.e. $\Pr[R_0=d_i^{(k)} | I_0=k] \sim c_k/i^{\alpha_k}, 1 \leq k \leq 10$. 
Let $N_k$ be the number of data items of flow~$k$. 
Set $\alpha_i=2.5, 1\leq i \leq 5$, $\alpha_j=1.5, 6\leq j \leq 10$, and $N_k=10^6, 1 \leq k \leq 10$, and
therefore, $c_i=(\sum_{x=1}^{N_1} x^{-\alpha_1})^{-1}=0.7454$, $c_j=(\sum_{x=1}^{N_6} x^{-\alpha_6})^{-1}=0.3831$. 
Note that we use the enhanced approximation~(\ref{eq:minverseaprxC3}), instead of~(\ref{eq:minverseC3}),  to compute $m^\leftarrow(x)$ when the cache size~$x$ is relatively small.
The theoretical and empirical results for the miss probabilities are plotted in Fig.~\ref{fig:exp2} when changing the cache size from 200 to 2000. 
Since flows~$1-5$ (respectively flows~$6-10$) have the same popularity distribution and the same miss ratio, we only plot flow~$1$ (respectively flow~$6$). 
\begin{figure}[h]
\vspace{-0.3cm}
\centering
\includegraphics[width=7.2cm]{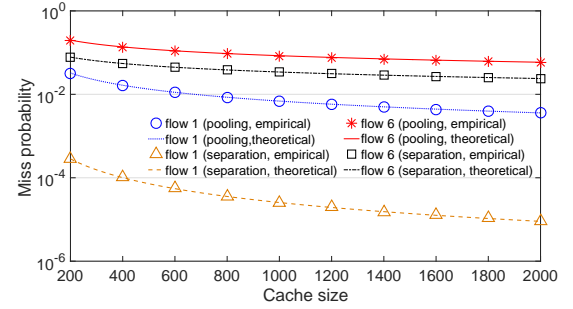}\vspace{-0.3cm}
\caption{Impacts among 10 flows}
\label{fig:exp2}
\vspace{-0.2cm}
\end{figure}
It can be observed that the empirical results match with the numerical results even when the cache size is relatively small.
In this case, flow~1 has a miss probability tail that decays on the order of $1/x^{0.9}$, as shown by the curve with a label flow 1 (pooling, empirical).  
However, if flow~1 is served without others, as shown by the curve with a label flow~1 (separation, empirical),  its probability tail only decays on the order of $1/x^{1.5}$. 
Therefore, in this case it is much worse for flow~1 (respectively flows $2-5$) to share with others than to be served exclusively. 
On the other hand,  flow~6 (respectively flows~$7-10$) is not significantly impacted when served together with other flows, since $\alpha_6<\alpha_i, 1\leq i \leq 5$.

\end{experiment}

\begin{experiment}\label{exp:3}
To address overlapped common data items, 
we simulate $3$ classes $A, B, D$ defined in~Section~\ref{ss:overlap}, 
and use the same notation introduced therein. 
Let $N_A, N_B, N_D$ be the numbers of data items of class $A, B$  and $D$, respectively.
\begin{figure}[h] 
\vspace{-0.2cm}
 \centering
 \subfigure[Without overlapped data (Corollary~\ref{corollary:4})]{
 \begin{minipage}{8.2cm}
 \centering
 \includegraphics[width=8.2cm]{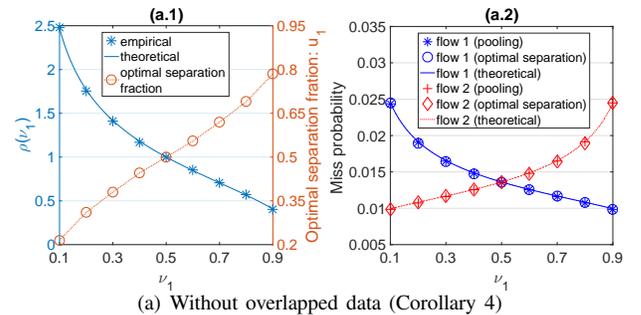}\vspace{-0.0cm}
 \end{minipage}
 }
 \vspace{-0cm}
  \subfigure[With overlapped data (Corollary~\ref{corollary:5})]{
 \begin{minipage}{6.9cm}
 \centering
  \vspace{-0.2cm}
  \includegraphics[width=6.9cm]{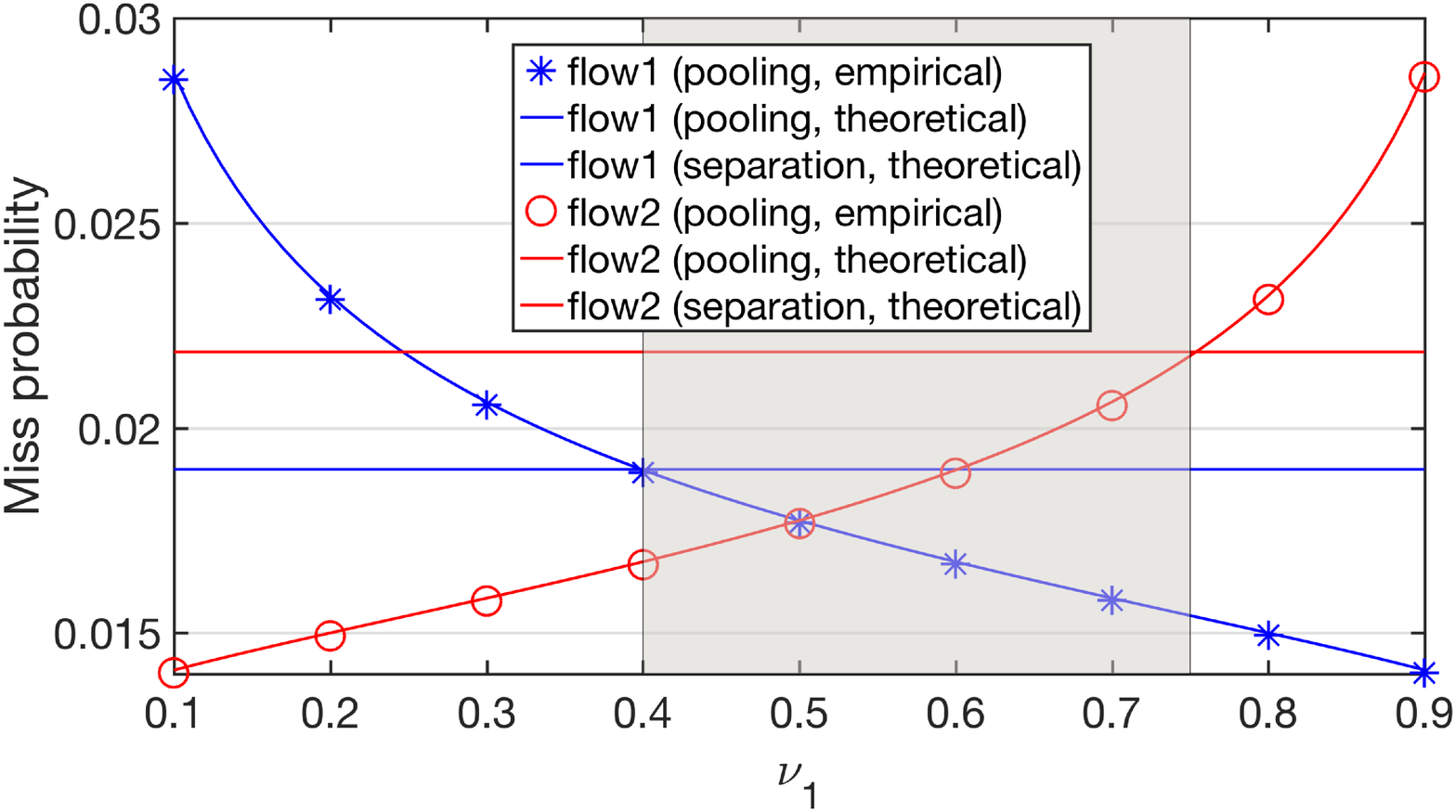}\vspace{-0.0cm}
 \end{minipage}
}
\vspace{-0.15cm}
\caption{Two flows sharing a server}
\label{fig:exp3}
\vspace{-0.25cm}
\end{figure}
Set cache size $x=1000$, $N_A=N_B=N_D=10^6$, $\alpha_A=\alpha_B=\alpha_D=1.7$, $c_A=c_B=c_D=\left(\sum_{i=1}^{N_A} i^{-\alpha_A}\right)^{-1}=0.4868$.
We conduct two experiments to study two flows with (respectively without) common data items by setting $p_D^{(1)}=p_D^{(2)}=0.2$ (respectively $p_D^{(1)}=p_D^{(2)}=0$).  
In Fig.~\ref{fig:exp3}(a), we plot $\rho(\nu_1) \eqdef \Pr_s^*[C_0>u_1x| I_0=1]/\Pr_s^*[C_0>u_2x| I_0=2]$ and the miss ratios for both flows under resource pooling and the optimal separation to validate Corollary~\ref{corollary:4}.
The simulations match with the theoretical results.
In Fig.~\ref{fig:exp3}(b), we plot the miss ratios under resource pooling and under a static separation $(u_1, u_2)=(0.55,0.45)$.
When flow~1 and flow~2 have common data and $\nu_1 \in (0.4, 0.75)$ (the shaded area in Fig.~\ref{fig:exp3}~(b)), 
both flows have lower miss ratios under resource pooling than under the static separation.
This result validates Corollary~\ref{corollary:5}. In presence of overlapped data, there exists a good region where both flows have better hit ratios by pooling.
However, when the arrival rates of these two flows are very different, i.e., $\nu_1<0.4$ or $\nu_1>0.75$,  the flow with a lower arrival rate will be negatively impacted.  

\end{experiment}